\documentclass[twocolumn,english,jmp,amsmath,amssymb,reprint,nofootinbib,nobalancelastpage,superscriptaddress]{revtex4-1}
\usepackage[T1]{fontenc}
\usepackage[latin9]{inputenc}
\usepackage{array}
\usepackage{multirow}
\usepackage{amsmath}
\usepackage{amsthm}
\usepackage{amssymb}

\makeatletter

\providecommand{\tabularnewline}{\\}

 \theoremstyle{definition}
 \newtheorem*{defn*}{\protect\definitionname}
\theoremstyle{plain}
\newtheorem{thm}{\protect\theoremname}
  \theoremstyle{plain}
  \newtheorem*{thm*}{\protect\theoremname}
  \theoremstyle{plain}
  \newtheorem{lem}[thm]{\protect\lemmaname}

\usepackage[table]{xcolor} 
\usepackage{colortbl}

\usepackage{tikz}
\usepackage{dsfont}

\usepackage{titlesec}
\titleformat*{\paragraph}{\bf\small}

\renewcommand{\Re}{\operatorname{Re}}
\renewcommand{\Im}{\operatorname{Im}}
\usepackage{array,booktabs}

\AtBeginDocument{\let\origref\ref
   \renewcommand{\ref}[1]{(\origref{#1})}}

\@ifundefined{showcaptionsetup}{}{%
 \PassOptionsToPackage{caption=false}{subfig}}
\usepackage{subfig}
\makeatother

\usepackage{babel}
  \providecommand{\definitionname}{Definition}
  \providecommand{\lemmaname}{Lemma}
  \providecommand{\theoremname}{Theorem}
\providecommand{\theoremname}{Theorem}

\begin{document}

\title{Continuous point symmetries in Group Field Theories}

\author{Alexander Kegeles}
\email{kegeles@aei.mpg.de}

\affiliation{Max Planck Institute for Gravitational Physics (Albert Einstein Institute),
~\\
Am Mühlenberg 1, 14476 Potsdam-Golm, Germany, EU}

\affiliation{Institute of Physics and Astronomy, University of Potsdam, ~\\
Karl-Liebknecht-Str. 24/25, 14476 Potsdam, Germany }

\author{Daniele Oriti}
\email{oriti@aei.mpg.de}

\affiliation{Max Planck Institute for Gravitational Physics (Albert Einstein Institute),
~\\
Am Mühlenberg 1, 14476 Potsdam-Golm, Germany, EU}
\begin{abstract}
We discuss the notion of symmetries in non-local field theories characterized
by integro-differential equations of motion, from a geometric perspective.
We then focus on Group Field Theory (GFT) models of quantum gravity
and provide a general analysis of their continuous point symmetry
transformations, including the generalized conservation laws following
from them.
\end{abstract}
\maketitle

\section*{Introduction}

Symmetry principles are omni-present in modern physics, and especially
in the quantum field theory formulation of fundamental interactions.
They enter crucially in the very definition of fundamental fields
and particles, they dictate the allowed interactions between them,
and capture key features of such interactions in terms of conservation
laws. They also offer powerful conceptual tools, for example, in the
characterization of macroscopic phases of quantum systems, as well
as many computational simplifications and powerful mathematical techniques
for numerical and analytical calculations.

Group Field Theory (GFT) \citep{Oriti:2006se,Oriti:2009wn,Oriti:2011jm,Oriti:2013aqa,Krajewski:2012aw,Oriti:2014uga}
is a promising candidate formalism for a fundamental theory of quantum
gravity. It can be seen as generalization of matrix and tensor models
for two and higher dimensional gravity in terms of random triangulations
\citep{Gross:1991hx,Gurau:2011xp,Rivasseau:2011hm,Carrozza:2013mna},
and as a 2nd quantized, quantum field theoretical reformulation of
Loop Quantum Gravity \citep{Oriti:2013aqa,Oriti:2014uga}. As such,
it is suitable for the adaptation of the standard symmetry analysis
techniques from ordinary quantum field theories. This is the goal
of the present article. At the same time, the intrinsically non-local
structure of GFTs requires non-trivial adaptation of the well-known
Lie group-based symmetry analysis and the generally curved nature
of the GFT base manifold makes the exact computations and the whole
analysis considerably more involved. We will tackle both types of
difficulties in the following, but it is the non-local character of
the theory that stands out as the truly defining challenge, so we
focus mainly on it, and limiting ourselves to the classical aspects
of the theory, which are interesting enough. We had laid some of the
mathematical foundations of the present analysis in a previous work
\citep{Kegeles:2015oua}, to which we will often refer in the following.

A possible definition of the non-locality in the case of GFT is the
integro-differential structure of the correspondent equations of motion.
From this perspective the symmetry analysis of these quantum field
theories could be mapped to the symmetry analysis of the integro-differential
equations of motion, at least in the classical approximation. The
Lie algebra approach to symmetries of integro-differential equations
is not new, and there is a quite extensive literature available on
this subject \citep{Zawistowski:2001bt,ibragimov2002symmetries,meleshko2010symmetries}.
However, in contrast to local field theories, the methods and strategies
of analysis strongly vary from case to case, and no universal method
is known yet. As a result, we do not know a priori which of the known
methods can be suitable adapted for the analysis of group field theories.
Particular examples for symmetry calculations for several GFT models
were calculated in \citep{BenGeloun:2011cz} but no systematic treatment
was proposed there.

In this paper we investigate the \textsl{variational} \textsl{symmetry}
groups of various prominent models in Group Field Theory, that is
symmetries of the action and not the (wider) symmetries of the equations
of motion. We consider only \textsl{point symmetries}, for several
interesting models, and investigate the consequences of the derived
symmetries in terms of generalized conservation laws, on the basis
of the framework performed in \citep{Kegeles:2015oua}. We will also
show that, when special types of matter fields are included in the
models, a possible definition of ``conserved'' charges can be derived
from the same generalized conservation laws. Also the restriction
to point symmetries will be motivated in the following, and we leave
the generalization to contact or Lie-Baecklund symmetries to be discussed
elsewhere. Similarly, we leave for further work the study of the consequences
of the symmetries we identify at the quantum level. While a detailed
analysis is of course needed, the main reason for postponing such
work is that we do not anticipate any additional difficulty (with
respect to the local quantum field theory case) in the derivation
of the corresponding Ward identities or in the study of possible anomalies.

The method we use in this paper is, to the large extent, the usual
Lie algebra calculations of symmetry groups, suitably modified to
adapt to the peculiar features of group field theories.

The presentation is organized as follows. We begin with a brief review
on Group Field Theories in section I. Then, after a recap of the definition
of various types of symmetries, in the geometric formulation of classical
field theories, we proceed with the symmetry analysis in Group Field
Theory in section II. In doing this, we clarify how to extend the
geometric treatment of field theories to the non-local case. At the
end of this section the reader can find the summarized table of symmetry
groups for the models under investigation. In section III we present
the correspondent Noether currents along with their generalized conservation
laws. The case, in which matter fields are present, is treated in
Section IV, with emphasis on scalar fields that can be used as relational
clocks, and the correspondent definition of relational charges.

\section{Group Field Theory }

In the first part of this section we give an informal definition of
the group field theory formalism and its general features, emphasizing
also the connections to other theories of gravity. In the second part
we introduce the specific models that we are going to study in the
rest of this paper. The notation and conventions introduced in this
section will be used throughout the paper.

A Group Field Theory \citep{Oriti:2009wn} is a quantum field theory
on a group manifold defined by a \textsl{non-local action} at the
classical level, and the corresponding path integral at the quantum
level. The fundamental fields of the theory are functions\textsl{
from a Lie Group} to some vector space (usually, the complex numbers)\footnote{In fact, the GFT formalism includes also the case in which the fields
are defined on a finite group, reducing to the \textsl{tensor models}
formalism. However, for obvious reasons our symmetry analysis would
not apply to this case, so we stick to the Lie group setting.}. In models directly related to lattice gravity and loop quantum gravity,
they are often assumed to satisfy a specific \textsl{gauge invariance}
condition, which indeed provides the perturbative Feynman amplitudes
of the model with a lattice gauge theory structure. This possibility
is, at its root, allowed by the choice of peculiar non-local pairing
of field arguments in the GFT interaction terms. In fact, this has
the immediate result that the perturbative expansion of the quantum
theory gives Feynman diagrams dual not just to graphs but to cellular
complexes, which can also be understood as discretization of some
smooth manifold. The Feynman amplitudes, whose explicit expression
is of course model-dependent, can be given the form of lattice gravity
path integrals \citep{Baratin:2010wi,Baratin:2011hp,Baratin:2011tx}
or, equivalently, spin foam models \citep{Baez:1997zt,Alexandrov:2010un,Geloun:2010vj,Alexandrov:2011ab,Perez:2012wv}.
The latter are a covariant definition of the quantum dynamics of spin
networks, the quantum states of Loop Quantum Gravity (LQG) \citep{Rovelli2008,Chiou:2014jwa}.
In fact, GFTs can be understood as a reformulation of the kinematics
and dynamics of LQG degrees of freedom in a 2nd quantized framework
\citep{Oriti:2013aqa}. The GFT Hilbert space re-organizes the same
type of spin networks in a Fock space, whose fundamental quanta are
spin network vertices and are created (annihilated) from (into) a
Fock vacuum (a state with no geometric nor topological structure)
by the action of the GFT field operators, absent any embedding information
into any ambient smooth continuum manifold.

In the following we will discuss each point in more detail, in order
to convey the general idea and motivation behind each of the above
features.

\paragraph{Non-local action.}

In local theories the action is an integral functional on the (appropriate)
space of fields, whose kernel is the Lagrangian. In the geometrical
interpretation a Lagrangian is a function on a vector bundle which
is usually a jet bundle over a principle $G$-bundle where $G$ is
the fundamental symmetry group of the theory. In usual field theories
of fundamental interactions, the base manifold of the vector bundle
is interpreted as space time and the fiber is a vector space that
carries a representation of $G$.

In non-local theories we want to maintain this geometrical picture
as far as possible, even if in our GFT context the base manifold will
not have the interpretation of spacetime (but is rather related to
\textsl{superspace}, the space of smooth spatial geometries, or minisuperspace,
the space of homogeneous spatial geometries; see \citep{Gielen:2013naa,Gielen:2014uga,Gielen:2014ila,Oriti:2016qtz}).
We assume that a non-local function can be treated as a local function
on a higher dimensional space. The drawback of this picture is that
different non-local therms are described by different geometrical
bundles, and the geometrical structure of the theory strongly depends
on the model in question, in contrast to the local case.

In a local theory the action is an integral of a Lagrangian over some
domain $\Omega$. In non-local theory it becomes a sum of integrals
whose domains are different base manifolds (in particular, of different
dimension). In the specific types of non-local field theories we will
be concerned with, they correspond to different numbers of copies
of a given base manifold. The action can therefore be generally written
as 
\begin{equation}
S=\sum_{i}\int_{M_{i}}\,L^{i}\,\text{vol}_{i},\label{eq:Non local action}
\end{equation}
where $i$ ranges over the number of different base manifolds $M_{i}$
with the correspondent Lagrangians $L^{i}$. The measure of integration
on each manifold defines the volume density, and in the GFT case,
will be given by the Haar measure or some other invariant measure,
depending on the Lie group chosen as base manifold.

\paragraph{Lie group structure of the base manifold.}

Indeed, in Group Field Theory, we require the local base manifold,
which matches by definition the domain of each individual GFT field,
to be given by some number of copies of a Lie group $G$. In GFT models
of quantum gravity, this is usually chosen to be the local gauge group
of gravity, i.e. the Lorentz group or it double cover $SL\left(2,\mathbb{C}\right)$
(or its Riemannian counterpart $Spin(4)$ for models of gravity in
Euclidean signature). For models directly related to Loop Quantum
Gravity, the rotation subgroup $SU\left(2\right)$ of the Lorentz
group becomes the relevant base manifold, via appropriate conditions
imposed on the GFT fields (and quantum states) at the dynamical level.

The number of copies of the group $G$ defining the local base manifold
is usually the topological dimension of the cellular complexes dual
to the Feynman diagrams of the model, chosen to match the dimension
of the continuum spacetime one aims to reconstruct from the quantum
dynamics of the model. Clearly, to have a physical connection to General
Relativity we need to develop and understand models in four dimensions.
However, since these tend to be naturally more involved than their
lower dimensional counterparts, several two and three dimensional
models have been studied, which allow to investigate important mathematical
and conceptual problems of the theory, in the setting with reduced
complexity.

\paragraph{Gauge invariance of fields}

In some GFT models the fields are required to satisfy a so-called
gauge invariance condition. Explicitly it means that for any $h$
in the diagonal subgroup $G_{D}=\left\{ \left(g,\cdots,g\right)\in G^{\times n}\vert g\in G\right\} $
the fields satisfy 
\begin{equation}
\phi\circ R_{h}=\phi,\label{eq:Gauge invariance of fields}
\end{equation}
where $R_{h}:G^{\times n}\to G^{\times n}$ denotes the right multiplication
by $h\in G_{D}$. Due to this gauge invariance condition, the base
manifold of the GFT fields is effectively reduced to a quotient of
$n$ copies of the group under the diagonal group action 
\[
G\times\cdots\times G/G_{D}.
\]
This condition is imposed for many different physical considerations.
In particular, in gauge invariant GFT models the perturbative Feynman
amplitudes of the theory take the form of lattice gauge theories (on
the cellular complex dual to each Feynman diagram) and the quantum
states become those of a lattice gauge theory with gauge group $G$
(in particular, for $G=SU(2)$, a complete orthonormal basis is given
by spin networks, and the same amplitudes can be equivalently written
as spin foam models).

\subsubsection*{Other defining features of different GFT models}

Beside the choice of base manifold, there are various other ingredients
that have to be specified, in order to fully define a GFT model, and
this even before one chooses a functional form for the interaction
kernels. The main differences between various models include: 

\paragraph{Presence (absence) of derivatives in the local Lagrangian - dynamical
(static) models}

The GFT models proposed at first, for a study of topological field
theories of BF type and, later, for 4d quantum gravity described as
a constrained BF theory and in absence of matter fields, only possessed
a ``mass'' term, and no derivatives of the fields in their quadratic,
local part of the Lagrangian. In local QFT, this would imply a trivial
dynamics, and therefore one could label these models \textquotedbl{}static\textquotedbl{}.
However, in contrast to local field theories, the transition functions,
and more generally both the classical and quantum dynamics in \textquotedbl{}static\textquotedbl{}
group field theories are still highly non-trivial due to their non-local
nature. More recently, GFT models which include derivative terms in
the quadratic part of the action have been studied extensively. They
can be motivated in various ways, the first being that renormalizability
seems to require them, at least for some models in which a non-trivial
dynamical term will be generated by the RG-flow \citep{Geloun:2011cy,Geloun:2013zka},
and therefore needs to be included in the theory space.

\paragraph{Non-local structure of the Lagrangian - combinatorics}

As said, the main feature distinguishing GFTs from ordinary local
field theories is the combinatorial pattern of relations between field
arguments in the GFT interactions. In principle many different non-local
interaction terms can be included in the action. A preference of one
model over another can be given conclusively only by extracting its
physical consequences. We will discuss the possible combinatorics
and their consequences rather extensively in the following. Here we
only stress that a detailed analysis of symmetries of the corresponding
models is going to be useful also for choosing one combinatorial structure
over another.

\paragraph{Number of different fundamental fields - colored theories}

One can also consider GFT models involving more than a single fundamental
field. If a Group Field Theory model involves more than one field,
we call the model \textsl{colored} and distinguish different fields
by an additional index, calling it the color index. Colored GFTs were
introduced for the first time in \citep{Gurau:2009tw} for models
aimed at describing simplicial quantum gravity, and topological BF
theories discretized on simplicial complexes. Indeed, this step immediately
led to a large number of interesting mathematical results and powerful
new techniques, in the GFT context as well as for the simpler tensor
models \citep{Gurau:2011xp}. In particular, Feynman diagrams generated
by non-colored simplicial GFT models can be dual to very singular
simplicial complexes, while Feynman diagrams of colored models are
much more regular, and their topology can be reconstructed to a much
greater extent \citep{Gurau:2010nd}.

\paragraph{Quantum statistics}

An additional assumption on the theory is its quantum statistics,
i.e. whether the GFT quanta of a given model are bosonic, fermionic
or of other nonstandard statistics. In local, spacetime-based quantum
field theories the quantum statistics is highly constrained by powerful
spin-statistics theorems, linking the quantum statistics to the spin
of the quanta, under the assumption of Lorentz or Poincare invariance
of the theory. No similar spin-statistics theorem is available, yet,
in the GFT framework. First and foremost, this is due to the fact
that the base manifold in group field theories is not directly associated
with space time, thus we have no obvious symmetry requirement to impose,
like Lorentz or Poincare invariance. Second, as we have already stressed
as a motivation for our work, very little is known about GFT symmetries,
from a full classification of them in specific models to their general
consequences, on the statistics of the same models and on their physics.
This paper is meant to partially fill this gap. 

\subsubsection*{Notation}

Throughout the paper we will denote an element of the Lie group $G^{\times n}$
as 
\begin{equation}
\vec{g}=\left(g_{1},g_{2},g_{3},\cdots,g_{n}\right).
\end{equation}
Differential operators with an index will refer to operators that
act on the correspondent copy of the group. For example, 
\begin{align}
\nabla_{1}\left(\vec{g}\right) & =\left(\nabla g_{1},g_{2},\cdots,g_{n}\right).
\end{align}
Differential operators with more than one index refer to a sum of
individual operators as 
\begin{equation}
\nabla_{123}=\nabla_{1}+\nabla_{2}+\nabla_{3}.
\end{equation}
If used without further clarification, an integral symbol (without
the explicit measure) denotes an integral over all variables that
appear under the symbol. Integration over each single group element
is performed with the Haar measure on $G$ 
\begin{equation}
\int\phi\left(g_{1},g_{2},g_{3}\right)=\int\text{d}g_{1}\text{d}g_{2}\text{d}g_{3}\,\phi\left(g_{1},g_{2},g_{3}\right).
\end{equation}
This notation will be used a bit differently in section III, where
we will point out the differences explicitly. The fields are complex
scalar fields. The upper script of the field denotes the color and
the subscript denotes the field's dependence on the variables 
\begin{equation}
\phi_{1,2,\cdots,n}^{c}:=\phi^{c}\left(g_{1},g_{2},\cdots,g_{n}\right).
\end{equation}

\subsection{Overview of the models discussed in the paper}

We are going to present the major distinctions of combinatorial structures
of models discussed in the following. The general structure of the
GFT actions has the form of \ref{eq:Non local action} as 
\begin{equation}
S\left[\phi\right]=S^{loc}\left[\phi\right]+S^{nloc}\left[\phi\right].
\end{equation}
We assume that the local, quadratic part of the action is defined
as 
\begin{equation}
S^{loc}\left[\phi\right]=\int_{G^{\times n}}\,\kappa\nabla\bar{\phi}\cdot\nabla\phi+m\,\bar{\phi}\phi,\label{eq:local part}
\end{equation}
where $\nabla$ is the gradient on the group $G$ and the $\cdot$
denotes the contraction of the two vectors. We will also treat cases
in which $\kappa$ is zero, meaning that the model is a static one.
In the following, we distinguish between three different types of
models - simplicial, tensorial and geometrical. The corresponding
interaction parts $S^{I}$ are presented below and a concise summery
of the interaction terms used is given in the table \ref{tab:Overview-of the models}.

\subsubsection{Simplicial}

The interaction part of simplicial models is constructed such that
the Feynman diagrams have a particular topological interpretation,
i.e. they are simplicial complexes. Let us illustrate in one example
how this comes about, in a simple example: a GFT model for BF theory
in three dimensions and the GFT fields subjected to the gauge invariance
condition. If the group is chosen to be $SU(2)$, this is a model
for 3d gravity in euclidean signature. In this case the simplicial
interaction is given by 
\begin{equation}
S^{I}\left[\phi^{c}\right]=\lambda\int_{\Omega}\,\phi_{1,2,3}\phi_{3,4,5}\phi_{5,2,6}\phi_{6,4,1}+c.c.\;,\label{eq:Boulatove usual}
\end{equation}
where the integral domain is $\Omega=G{}^{\times6}$. If we associate
to each field in the interaction a triangle with edges labelled by
the three arguments of the same field, the interaction will be associated
to a tetrahedron (a 3-simplex) formed by the 4 triangles glued pairwise
along common edges, as shown in figure \ref{fig:Geometrical-interpretation-of}.

\vspace{-1cm}
\begin{figure}[h]
\subfloat[3 dim field associated with a triangle\label{fig:Field-in-three-a}]{\begin{tikzpicture}[scale = 0.5]
\coordinate (0) at (0,0);
\node[circle, inner sep=1pt, fill=black] at (0) {};
\foreach \a in {90,210,330}{
\coordinate (\a) at (\a: 2.2cm);
}

\filldraw[draw=gray!80!white,fill=gray!20!white] (90) -- (210) -- (330)--cycle;

\begin{scope}[very thick,gray]
\draw (0) -- (30:1.4cm);
\draw (0) -- (150:1.4cm);
\draw (0) -- (270:1.4cm);
\end{scope}
\node[circle,fill=gray!20!white]  at(0) {$\phi$};
\end{tikzpicture}

}\hspace*{1cm}\subfloat[Simplicial interaction in three dimensions associated with a tetrahedron\label{fig:Simplicial-interaction-inb}]{\begin{tikzpicture}[scale = 0.5]
\node[circle, inner sep=2pt, fill=black,gray] at (0) {};
\coordinate (0) at (0,0);
\node[circle, inner sep=2pt, fill=black] at (0) {};
\foreach \a in {90,210,330}{
\coordinate (\a) at (\a: 2.2cm);

}


\path (0) edge[thick, dashed, gray]
	(90) edge[bend right=20, gray, thick , dashed] (210)
			edge[bend left=20, thick, gray, dashed] (330);

\filldraw[draw=gray,fill=gray!30!white, opacity =0.5] (0,-1) -- (2,1) -- (-2,1)-- cycle;
\filldraw[draw=gray,fill=gray!60!white, opacity =0.5] (0,-1) -- (-2,1) -- (0,-3)-- cycle;
\filldraw[draw=gray,fill=gray!100!white, opacity =0.5] (0,-1) -- (0,-3) -- (2,1)-- cycle;

\path (90) 	edge[bend right, very thick,gray] (210)
				edge[bend left, very thick,gray] (330);
\path (210) edge[bend right=40, very thick,gray] (330);

\node[anchor=south west] at (0) {$\textcolor{gray}{\phi}$};
\node[anchor=south] at (90) {$\phi$};
\node[anchor=east] at (210) {$\phi$};
\node[anchor=west] at (330) {$\phi$};

\foreach \a in {90,210,330}{
\node[circle, inner sep=2pt, fill=black,gray] at (\a) {};
}
\end{tikzpicture}

}

\caption{Topological interpretation of the field and the simplicial interaction
in three dimensions\label{fig:Geometrical-interpretation-of}}
\end{figure}
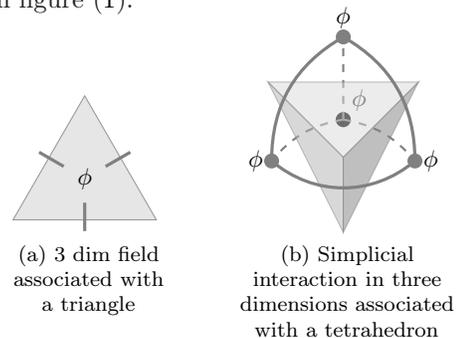

The Feynman diagrams will similarly be in correspondence with the
simplicial complexes obtained by gluing the different tetrahedra associated
to the interaction vertices in the Feynman diagram, along shared triangles,
the gluing being identified by a propagator line. Depending on the
details of the Feynman diagram, the resulting simplicial complex may
or may not be a simplicial manifold and can be quite singular. Moreover,
the data present in the GFT diagram are not, in general, sufficient
for reconstructing the topology of the same simplicial complex in
its entirety. Such technical difficulties are cured by introducing
colored fields \citep{Gurau:2010nd}.

The extension to higher dimensions, via extension of the base manifold
of the fields, and appropriate pairing of their arguments in the interactions,
follows the same criteria and it is straightforward.

\subsubsection*{Remarks on the combinatorial structure of simplicial models}

The original Boulatov model \citep{Boulatov:1992vp} for 3d gravity
has a slightly different combinatorial structure from the one we introduced
above. It is given by 
\begin{equation}
S^{I}\left[\phi^{c}\right]=\lambda\int_{\Omega}\,\phi_{1,2,3}\phi_{1,4,5}\phi_{2,5,6}\phi_{3,6,4}+c.c.,
\end{equation}
with an additional invariance of fields under cyclic permutations
of the variables. The quantum geometric content of the model is not
affected, as it can be seen both in the group representation, and
in the spin representation. In fact, in the Peter-Weyl decomposition
both combinatorial structures lead to a 6J symbol, which encodes both
the gauge invariance properties and the piece-wise flatness of the
simplicial complex generated in the perturbative expansion. However,
while the original Boulatov model produces a usual 6J symbol, the
interaction pattern we introduced above produces an additional factor
$\left(-1\right)$ that alternates with the representations involved.

In the case of colored models this difference can be absorbed in the
redefinition of the fields as $\tilde{\phi}^{i}=\phi^{i}\circ P^{i}$,
where $P^{i}$ being some permutation of the group elements. In this
way the exact order of the variables in the field becomes unimportant.
We will choose the following combinatorics, since it leads, as we
will show, to the largest symmetry group: 
\begin{equation}
S^{I}\left[\phi^{c}\right]=\lambda\int_{\Omega}\,\phi_{1,2,3}^{1}\phi_{1,4,5}^{2}\phi_{6,2,5}^{3}\phi_{6,4,3}^{4}+c.c..
\end{equation}
It is important to mention that if the models are indeed equivalent,
their symmetry group should not differ as well. This implies that
a particular choice of the combinatorics may simply help to discover
symmetries that would still be there for different combinatorics,
but would be more difficult to identify. In the next section we will
show how these minor combinatorial differences affect the symmetry
group.

Notice that, while in the action above we have chosen four GFT fields
to appear in one term, with their complex conjugates appearing in
the other, our focus here was only the combinatorial structure, and
one can devise simplicial interactions involving both the field and
its complex conjugate in the same monomial. For example, we can start
with the action from equation \ref{eq:Boulatove usual} and color
the fields in the way $\phi^{1}=\phi^{3}=\phi$ and $\phi^{2}=\phi^{4}=\bar{\phi}$
such that the interaction part coincides with its complex conjugate
\begin{equation}
S^{I}\left[\phi^{c}\right]=\lambda\int_{\Omega}\,\phi_{1,2,3}\bar{\phi}_{3,4,5}\phi_{5,2,6}\bar{\phi}_{6,4,1}.
\end{equation}
In this case we refer to the above action as colored, with two colors,
even though the model involves only the field $\phi$ and its complex
conjugate. This convention will become handy in the classification
of the symmetries in the following. 

\subsubsection{Tensorial models}

Tensor models are characterized by an $U\left(N\right)$ invariance.
Given a rank-n complex tensor $T_{i_{1}...i_{n}}$ with index set
of dimension $N$, it transforms naturally under the group $U(N)^{\times n}$,
where $U(N)$ is a unitary $N\times N$ matrix, acting on each of
its indices. This is also the natural symmetry of tensor interactions,
so that the full theory space is defined to be spanned by all possible
monomials in the tensor and its complex conjugate, with their indices
contracted to give unitary invariants \citep{Gurau:2011xp,Bonzom:2012hw}.
The generalization of the same invariance characterizes the interactions
of tensorial GFTs \citep{Carrozza:2013mna,Benedetti:2014qsa}. Hereby
a monomial in fields belongs to the theory space if it is invariant
under a unitary transformation defined as follows 
\[
U\phi\left(\vec{g}\right)=\int\text{d}\vec{h}\,U^{1}\left(g_{1},h_{1}\right)U^{2}\left(g_{2},h_{2}\right)U^{3}\left(g_{3},h_{3}\right)\phi\left(\vec{h}\right),
\]
with the requirement on the kernels $U^{i}$ to satisfy 
\begin{equation}
\int_{\Omega}\text{d}h\,U\left(g,h\right)U^{\dagger}\left(h,q\right)=\delta\left(g,q\right),
\end{equation}
where $U^{\dagger}\left(h,g\right):=\bar{U}\left(g,h\right)$. This
conditions requires that two fields which share a group element need
to be complex conjugate of each other. It is easy to verify that this
excludes the simplicial combinatorics. It is also important to mention
that the kernels $U^{i}$ do not need to be smooth, differentiable
or even continuous and for this reason they may include delta distributions.We
will come back to this point in the next section, when we discuss
the symmetries of the tensorial models.

Note that a dynamical term will in general break the unitary invariance.
Therefore, when we refer to tensorial dynamical models in the following,
we imply that the unitary invariance characterizes only the interaction
part and not of the whole action. Note also that one can have a very
similar type of invariance for real GFT fields, with the unitary group
replaced by an orthogonal group. The construction proceeds in analogous
way.

\begin{table*}[t]
\begin{centering}
\renewcommand*{\arraystretch}{2}%
\begin{tabular}{|>{\raggedright}m{0.2\textwidth}|>{\centering}m{0.15\textwidth}|c|>{\centering}m{0.15\textwidth}|c|}
\hline 
 & \multicolumn{2}{>{\centering}m{0.3\textwidth}|}{\textbf{Gauge variant}} & \multicolumn{2}{>{\centering}m{0.3\textwidth}|}{\textbf{Gauge invariant}}\tabularnewline
\hline 
\textbf{Dimension}  & \multicolumn{4}{c|}{$3D$ and $4D$}\tabularnewline
\hline 
\multirow{2}{0.2\textwidth}{\textbf{Kinetic part}} & \multicolumn{4}{c|}{$\kappa\,\nabla\phi\cdot\nabla\phi+m\,\phi\phi$}\tabularnewline
\cline{2-5} 
 & \multicolumn{2}{c|}{$\kappa\neq0$} & $\kappa\neq0$  & $\kappa=0$\tabularnewline
\hline 
\textbf{Group}  & $SU\left(2\right)$  & $SU\left(2\right)$  & $SU\left(2\right)$  & $G$\tabularnewline
\hline 
\textbf{Combinatorics}  & Tensorial  & Simplicial  & \multicolumn{2}{c|}{Simplicial}\tabularnewline
\hline 
\textbf{Colors}  & -  & (un)colored  & \multicolumn{2}{c|}{(un)colored}\tabularnewline
\hline 
\textbf{Simplicity constrains }  & -  & -  & -  & Barrett-Crane\tabularnewline
\hline 
\end{tabular}
\par\end{centering}
\centering{}\caption{Overview of the models discussed in this paper\label{tab:Overview-of the models}}
\end{table*}

\subsubsection{Extended Barrett-Crane model}

In four and higher dimensions, gravity can be formulated as a BF theory
plus appropriate constraints \citep{Freidel:1999rr}, which are labeled
\textsl{simplicity constraints}. This goes under the name of Plebanksi
formulation of gravity. This formulation provides also the conceptual
and technical starting point for the construction of spin foam and
group field theory models for 4d quantum gravity. One may call the
corresponding GFT models \textsl{geometric}, even though one has a
direct control only on the discrete (simplicial) geometric interpretation
of states and amplitudes, while the reconstruction of continuum geometry
requires more work. As an example of these constructions, we deal
with the so-called Barrett-Crane model \citep{Barrett:1997gw}, whose
detailed treatment in the language of extended Group Field Theory
was presented in \citep{Baratin:2011tx}, for the euclidean signature.
Here we show just the main features of the model and refer to the
cited literature for more details.

The starting point is the GFT model for 4d BF theory based on simplicial
interactions, in which the fundamental GFT field is associated to
a tetrahedron in 4d, and the interactions involve five GFT fields,
paired to represent the gluing of five tetrahedra to form a 4-simplex.
The base group manifold of the model is $Spin\left(4\right)$.

Simplicity constraints are characterized by a vector in $S^{3}\simeq SU\left(2\right)$,
interpreted as a unit normal vector (in $\mathbb{R}^{4}$) of the
tetrahedron represented by the field $\phi$. In order to keep track
of this additional normal vector the local base manifold is extended,
so that the field becomes a function on four copies of $Spin\left(4\right)$
and one copy of $SU\left(2\right)$ 
\begin{equation}
\phi\left(g^{1},g^{2},g^{3},g^{4},k\right)=:\phi_{1,2,3,4,k},
\end{equation}
where $g_{i}\in Spin\left(4\right)$ and $k\in SU\left(2\right)$.

The interaction of the model becomes an extended version of the Ooguri
interaction given as

\begin{eqnarray}
S^{I}\left[\phi\right] & = & \int\phi_{1,2,3,4,k_{1}}\phi_{4,5,6,7,k_{2}}\phi_{7,3,8,9,k_{3}}\phi_{9,6,2,10,k_{4}}\phi_{10,8,5,1,k_{5}}\nonumber \\
 &  & +\,c.c..\label{eq:BC interaction}
\end{eqnarray}
The gauge invariance is again written in the usual form as 
\begin{equation}
\phi\circ R_{h}=\phi,
\end{equation}
with $h\in Spin\left(4\right)_{4D}$. Additionally, the simplicity
constraints are imposed by requiring invariance of the fields 
\begin{equation}
\phi\circ S=\phi,
\end{equation}
under the transformation 
\begin{equation}
S:\left(\vec{g},k\right)\mapsto\left(\vec{g};k\right)\cdot\left(\left(k\vec{u}k^{-1},\vec{u}\right);\mathds{1}\right),
\end{equation}
where $u^{j}\in SU\left(2\right)$. If we write a $Spin\left(4\right)$
element in its selfdual and anti-selfdual $SU\left(2\right)$ components
as $g=\left(g_{-},g_{+}\right)$, the above transformation takes the
form 
\begin{equation}
S\left(\vec{g};k\right)=\left(g_{-}^{j}\,ku^{j}k^{-1},g_{+}^{j}u^{j};k\right).\label{eq:Simplicity constraint}
\end{equation}
In \citep{Baratin:2011tx} it has been shown that $S$ and $R_{h}$
commute as projectors acting on the space of fields, which allows
to combine them into a single transformation, which is itself a projector,
acting on the GFT fields as

\begin{eqnarray*}
\mathcal{S}: & \left(\vec{g};k\right) & \mapsto\left(\mathds{1};h_{-}^{-1}\right)\cdot\left(\vec{g};k\right)\cdot\left(\left(k\vec{u}k^{-1},\vec{u},\right);\mathds{1}\right)\cdot\left(h_{-};h_{+}\right),
\end{eqnarray*}
where $h\in Spin\left(4\right)$. One can indeed verify that the fields
invariant under the above transformation satisfy 
\begin{equation}
\phi\circ\mathcal{S}=\phi.
\end{equation}

Notice that, since the simplicity and gauge invariance conditions
are imposed on the fields via a projector, the imposition of these
conditions on all fields appearing in the action is the most natural
choice, but any other choice, e.g. imposing them only on the fields
appearing in the interactions, would result in the same Feynman amplitudes
(but not the same theory, as for example the classical equations of
motion would be different). This is not true for other 4d gravity
models, where the simplicity constraints take a different form \citep{Baratin:2011hp}.

\section{Symmetries}

In this section we will present the different notions of symmetry
transformations in non-local field theories in general (recalling
the geometric construction in the local case, first), and then apply
them to Group Field Theory in particular, and derive the symmetry
groups for the models introduced above, showing the main steps of
the calculations for three-dimensional models.

\subsection{Transformations of local field theory}

The geometrical construction of local field theory is very well known,
but we will briefly review its main points here because they will
be essential in the following discussion of the non-local case.

In the geometrical picture, the Lagrangian is a differentiable (in
a sense that needs to be further specified) function on an $n$th
order jet bundle. In order to bring the main idea across without complicating
it with technical details, we will assume that the Lagrangian is a
function just on a vector bundle. The full construction can be found
in usual text books on this subject some of which are \citep{olver2000applications,meleshko2010symmetries}.

We call the relative vector bundle $E$, the base manifold of $E$
being $M$ and the fiber being $\mathbb{V}$. Locally, we can think
about $E$ as a cross product of $M\times\mathbb{V}$, which we assume
for the rest of this discussion. The points on $E$ are then given
by $x\in M$ and $u\in\mathbb{V}$, we write $\left(x,u\right)\in E$.
Hence, the values of the Lagrangian can be denoted as $L\left(x,u\right)\in\mathbb{R}$.

We then introduce the physical fields $\phi$ in the construction.
This is done by choosing points of the vector bundle which are given
by a smooth section of $E$. In other words we assume that $u=\phi\left(x\right)$.

Assuming that the set of transformations of the theory forms a group
$G_{T}$, we can write the action of the group on $E$ as 
\begin{equation}
g\cdot\left(x,u\right)=\left(\tilde{x},\tilde{u}\right)=\left(C\left(x,u\right),Q\left(x,u\right)\right).\label{eq:Lie point trafos}
\end{equation}
The action is thus specified by two functions $C$ and $Q$. Note,
that in general both functions depend on $x$ and $u=\phi\left(x\right)$
and are not invertible. However, locally around each point of the
bundle these transformations are diffeomorphisms, due the the fact
that they represent an action of a Lie group.

We ask for the transformed sections $\tilde{\phi}$ that corresponds
to a new point of the bundle, that is $\left(\tilde{x},\tilde{\phi}\left(\tilde{x}\right)\right)=\left(\tilde{x},\tilde{u}\right)$.
The transformed fields $\tilde{\phi}$ can then be seen as transformed
sections under the group action of $g$. It is a well known result
that the transformed fields are given by 
\begin{equation}
\tilde{\phi}\left(\tilde{x}\right)=Q\left(C^{-1}\left(\tilde{x}\right),\phi\circ C^{-1}\left(\tilde{x}\right)\right),
\end{equation}
or, in short, 
\begin{equation}
\phi\overset{g}{\mapsto}Q\circ\phi\circ C^{-1},\label{eq:trafo of the fields}
\end{equation}
at least as long as $C$ is invertible. Hereby, the transformation
$Q$ is defined along the fiber and the $C^{-1}$ accounts for the
transformation of the base manifold. 

We summarize the main properties of the maps $Q$ and $C$ before
finishing this part. For a fixed $\phi$ the base manifold transformation
$C$ is a local automorphism 
\begin{eqnarray}
C_{\phi}: & M & \to M\\
 & x & \mapsto C\left(x,\phi\left(x\right)\right).\nonumber 
\end{eqnarray}
And for a given point $x\in M$, the fiber transformation $Q$ is
a local automorphism 
\begin{eqnarray}
Q_{x}: & \mathbb{V} & \to\mathbb{V}\\
 & \phi\left(x\right) & \mapsto Q\left(x,\phi\left(x\right)\right).
\end{eqnarray}

\subsection{Transformations of a non-local field theory}

We now apply this construction to non-local field theories. As we
have pointed out earlier, the action is given by the sum of integrals
over Lagrangians 
\begin{equation}
S=\int_{M_{i}}\,L^{i}.
\end{equation}
Hereby, each of the Lagrangians is a function from a vector bundle
$E_{i}$ to $\mathbb{R}$. For $i\neq j$ the vector bundles $E^{i}$
and $E^{j}$ are assumed to be different. If they are not, we can
combine the Lagrangians $L^{i}$ and $L^{j}$ into a single Lagrangian
$L^{ij}=L^{i}+L^{j}$.

Following the general construction from the previous section we define
a transformation of the theory as transformation of the corresponding
vector bundles. Nevertheless, in the non-local case we need to transform
different bundles, which is why we say that a group action is given
by functions $C^{i},Q^{i}$ such that, for each $i$, $C^{i}$ and
$Q^{i}$ are transformations of $E^{i}$ in the above (local) sense.
It is important to realize that these transformations can not be independent
from one another, since they represent the same transformation $g\in G_{T}$.
Instead, their mutual relations should be given by the relation between
different vector bundles $E^{i}$.

Assuming that $E^{0}$ denotes the vector bundle, whose sections are
identified with physical fields $\phi^{0}=\phi$ we can quite generally
write each $E^{i}$ as a pull back of $n_{i}$ copies of $E^{0}$
by some embedding $f_{i}:M^{i}\to M^{\times n_{i}}$. These functions
$f^{i}$ encode the combinatorial structure of non-local Lagrangians
and provide a relation between different $E^{i}$'s. Therefore they
also give the relations between the sections $\phi^{i}$ as $\phi^{i}=\left(\phi^{0}\right)^{\times n_{i}}\circ f$.
Knowing how the field $\phi^{0}$ transforms under $g$ automatically
implies the transformations of $\phi^{i}$ as 
\begin{equation}
Q^{i}\circ\phi^{i}\circ\left(C^{i}\right)^{-1}=\left[Q^{0}\circ\phi^{0}\circ\left(C^{0}\right)^{-1}\right]^{\times n_{i}}\circ f.\label{eq:trafo of the sections}
\end{equation}
This relation implicitly defines $C^{i}$ as 
\begin{equation}
f\circ C^{i}=\left(C^{0}\right)^{\times n_{i}}\circ f,\label{eq:Relation between C's}
\end{equation}
and $Q^{i}$ as 
\begin{equation}
Q^{i}\circ\phi^{i}=\left[Q^{0}\circ\phi^{0}\right]^{\times n_{i}}.
\end{equation}

The above equations provide the missing link between the group actions
on different vector bundles. However, equation \ref{eq:Relation between C's}
does not always define a local automorphism $C^{i}$ on $M^{i}$.
If $C^{i}$ were an automorphism, equation \ref{eq:Relation between C's}
would imply that, for any $x^{i}\in M^{i}$, there exists an $\tilde{x}^{i}\in M^{i}$
such that 
\begin{equation}
f\left(\tilde{x}^{i}\right)=\left(C^{0}\right)^{\times n_{i}}f\left(x^{i}\right).\label{eq:condition on the group action}
\end{equation}
This however, is not always possible as we will see in the following
section.

If equation \ref{eq:condition on the group action} is not satisfied,
the group action can not be chosen consistently as a transformation
of the vector bundle. In this case, we can define the action of the
group directly on the space of fields, using equation \ref{eq:trafo of the sections}
as 
\[
g\cdot\left(x,u\right)=\left(x,Q\left(x,u,u_{x},\cdots\right)\right)=\left(x,\phi\circ C^{-1}\left(x\right)\right).
\]
Here, $u_{x}$ denotes the coordinates of the Jet space and refers
to derivatives of fields at the point $x$, i.e. $u_{x}=D\phi\vert_{x}$.
That the transformation $Q$ needs to depend on the derivatives of
fields $\phi$ is easily seen from the Taylor expansion, since 
\begin{align}
\phi\circ C^{-1}\left(x\right) & =\phi\left(x\right)-D\phi\left(X_{M}\right)+\mathcal{O}\left(X_{M}^{2}\right)\\
 & =u-u_{x}\cdot X_{M}+\mathcal{O}\left(X_{M}^{2}\right),
\end{align}
where $X_{M}$ is the infinitesimal generator of the transformation
$C$. Such transformations $Q$ generalize the notion of point transformations
from equation \ref{eq:Lie point trafos} to the so called Lie-Baecklund
transformations, which are transformations from the Jet bundle to
the vector bundle. In the next section, we will briefly explain how
the Lie-Baecklund transformations represent a more general notion
of symmetry. However, for reasons that will also become apparent in
the next section, we will restrict our analysis to Lie point symmetries.

\subsection{Notions of symmetry transformations}

There are many different notions of a continuous symmetry in local
field theories. Almost all of them are formulated as diffeomorphisms
of the vector bundle of the theory. In order to distinguish between
different notions of symmetries we first point out that an action
(in local theories) is defined as an integral, and therefore intrinsically
depends on the domain of integration over which the Lagrangian is
integrated, i.e. $S_{\Omega}\left[\phi\right]=\int_{\Omega}\,L$.
In this sense we can talk about a family of actions $\left\{ S_{\Omega'}\right\} $
for all $\Omega'\subseteq\Omega$. In the discussion of symmetries,
the dependence of the action on the domain plays a very important
role, which we are going to highlight in the following.

\subsubsection{Point symmetries}

The simplest notion of a symmetry of an action is a diffeomorphism
on the vector bundle of the theory \citep{olver2000applications,lie1880theorie}.
As we discussed above choosing a section of the bundle (physical field
$\phi$) it is possible to locally project the diffeomorphism on the
fiber and the base manifold obtaining the transformation function
$Q$ and $C$, which define a transformation of the fields and of
the base manifold respectively. A symmetry is then a transformation
which does not change the action functional $S_{\Omega'}$ for any
subdomain $\Omega'\subset\Omega$ 
\begin{equation}
S_{\Omega'}\left[\phi\right]\to S_{C\left(\Omega'\right)}\left[\tilde{\phi}\right]=S_{\Omega'}\left[\phi\right]\quad\forall\Omega'\subset\Omega\quad.
\end{equation}
These transformations are of the type \ref{eq:Lie point trafos} and
are called \textsl{Lie point symmetries} or ``\textsl{geometrical}''
symmetries, because they admit a geometrical interpretation of a flow,
being generated by vector fields on the vector bundle of the theory.

The requirement that the symmetry does not change the action for any
sub domain $\Omega'$ is essential, in order to be able to make point-wise
statements, i.e. to derive truly local statements from the existence
of the symmetry itself. In the physical literature, this statement
is often referred to as the Noether theorem, which allows the derivation
of point-wise equations of the form $\text{div}\left(J\right)=E_{L}\cdot\delta\phi$,
i.e. the conservation laws of the corresponding field theory. 

\subsubsection{Generalized symmetries}

A generalization of the symmetry concept (already introduced by Noether
in her original paper \citep{Noether1971aaa}) leads to the so called
\textsl{generalized symmetries}, which are Lie-Baecklund transformations
that can change the action by an arbitrary divergence term. That is,
for all $\Omega'\subset\Omega$ 
\begin{equation}
S_{\Omega'}\left[\phi\right]\to S_{C\left(\Omega'\right)}\left[\tilde{\phi}\right]=S_{\Omega'}\left[\phi\right]+\int_{\Omega'}\,\text{div}\left(\Gamma\right)\quad\forall\Omega'\subset\Omega\quad.
\end{equation}
A restriction of such transformations to those that can depend at
most on the first order derivatives of fields defines the so called
contact symmetries.

In general, a set of Lie-Baecklund symmetries is infinitely large,
but it is often the case that also infinitely many such transformations
are equivalent, leading to a finite number of inequivalent transformations.
Computational algorithms for finding Lie-Baecklund symmetries to a
fixed order of derivative dependencies are known and are implemented
in a large variety of computer algebra programs \citep{fushchich1989computer}.
Nevertheless, already for flat base manifolds (and of course, local
theories) the explicit calculations are quite challenging.

The reason for looking for a generalized notion of symmetry is the
observation that two actions are physically equivalent \textsl{if
and only if} they differ by a divergence term \citep{olver2000applications},
because the correspondent equations of motion are the same, and this
is all that matters in the classical regime. This implies that the
physically relevant object is not the action but rather an equivalence
class of actions\footnote{Is important to distinguish between the symmetries of the action and
symmetries of the correspondent equations of motion (which correspond
to extrema of the same action): generalized variational symmetries
form a subgroup of Lie-Baecklund transformations of the equations
of motion.}.

As in the previous cases, this class of symmetries gives rise to local,
point-wise equations. Even more, only in this case, the correspondence
between symmetries and divergence-free quantities like $\text{div}\left(J+\Gamma\right)=E_{L}\cdot\delta\phi$
is one to one, which is the actual statement of the original Noether
theorem.

\subsubsection{Integral symmetries}

An entirely different notion of symmetry arise if we drop the requirement
that a symmetry transformation should leave the family of actions
$\left\{ S_{\Omega'}\right\} $ invariant, and instead require the
invariance of $S_{\Omega}$ only for a single, fixed integral domain
$\Omega$, 
\begin{equation}
S_{\Omega}\left[\phi\right]\to S_{C\left(\Omega\right)}\left[\tilde{\phi}\right]=S_{\Omega}\left[\phi\right].
\end{equation}
This kind of transformations does not lead to point-wise statements,
like conservation laws. Clearly every symmetry of the previous type
is also a symmetry of this type but not the other way around.

It is interesting to observe that tensor models and tensorial GFTs
invoke exactly this type of symmetries in order to define the theory
space, since we require that the corresponding unitary transformations
satisfy 
\begin{equation}
\int_{G}\text{d}h\,U\left(g,h\right)U^{\dagger}\left(h,s\right)=\delta\left(gs^{-1}\right),
\end{equation}
only after the integration over the whole group $G$, and there is
no reason to assume that changing the integral domain to a subspace
of $G$ would preserve the above equality.

\subsubsection{Symmetries in non-local field theories}

So far the definition of a symmetry was introduced for an action which
is given by an integral over a Lagrangian. In the non-local case,
as we have explained, the action is given by a sum of such actions
each defined on a different base manifold. Therefore, we need to extend
the above notions of symmetries to transformations, which are symmetries
(in one of the above senses) of each and all individual functions
in the action-sum.

This is the only generalization we need, to start analyzing symmetries
of the non-local GFT models introduced above.

However, it is important to stress here, that this is not enough to
study generalized symmetries of the Lie-Baecklung type. The same motivation
that lead to considering them in the local case would apply as well
for non-local models. However, contrary to the local case, for non-local
field theories the equivalence class of actions that yields the same
equations of motion is not under control (to our knowledge). The only
thing that we can say is that it does not coincide with the one defined
in the local case, because as we pointed out in \citep{Kegeles:2015oua},
a divergence term will, in general, change the equations of motion.

One approach to overcome this difficulty would be to discuss directly
the symmetries of the corresponding equations of motion, which are
integro-differential equation (see \citep{meleshko2010symmetries}
for the standard approach of Lie algebra methods in integro-differential
equations). However, this is highly non-trivial, in the GFT case.
Gauge invariance condition, the structure of the curved base manifold,
as well as its large dimension make the usual Lie algebra approach,
even more involved.

Also, there is not much more to say about the integral symmetry transformations
defining tensorial group field theories, beside what we remarked already,
i.e that they provide a natural characterization of the corresponding
theory space. For these reasons, we limit our analysis to the Lie
point symmetry analysis and postpone the analysis of Lie-Baecklund
symmetries to future work.

In the following we will use the definition of a symmetry for a non-local
action as follows: 
\begin{defn*}
A symmetry of a non-local action is a transformation that is a Lie-point
symmetries of each functional in the action-sum. 
\end{defn*}

\subsection{Symmetry analysis of gauge variant models}

We start by performing the standard Lie group analysis of point symmetries
\citep{olver2000applications} in the case of gauge-variant GFT models.

In \citep{Kegeles:2015oua} we have shown, that a symmetry condition
of the action can be equivalently formulated on the level of its Lagrangians,
leading to a generalized version of Noether theorem, with respect
to the local case. More precisely, the symmetry relation can be formulated
in the following way 
\begin{thm}
$G$ is a symmetry group of the action iff the generators of the symmetry
$\left(X_{V},X_{M}\right)$ satisfy the relation 
\begin{align}
0= & D_{J}L\cdot DX_{Q}+D_{V}L\cdot X_{Q}+\text{Div}\left(L\,X_{M}\right),\label{eq: Local symmetry condition}
\end{align}
where $X_{Q}=X_{V}-X_{M}\left(\phi\right)$ and it is assumed that
every term is evaluated at some point $z$ of the correspondent base
manifold. 
\end{thm}
The notation that is used in the above equation needs to be further
explained: 
\begin{itemize}
\item $X_{M}$ is a vector field on the base manifold which coefficients
depend on a point of the base manifold and a point on the fiber. In
local coordinates $\left(U,x\right)$ of the base manifold the vector
field can be written as $X_{M}=X_{M}^{i}\left(x,\phi\left(x\right)\right)\,\partial_{i}$. 
\item $X_{V}$ is a vector field along the fiber of the bundle that in local
coordinates $\left(U,x\right)\times\left(V,u\right)$ can be denoted
as $X_{V}^{i}=X_{V}^{i}\left(x,\phi\left(x\right)\right)\,\partial_{u^{i}}$
where $u^{i}=\phi^{i}\left(x\right)$. We will sometimes use the simpler
notation $\partial_{\phi\left(x\right)}$ or even $\partial_{\phi}$,
always referring to $\partial_{u}$. 
\item The assumption of dealing with a geometrical symmetry translates into
the restriction of the coefficients of vector fields depending only
on $x$ and $\phi\left(x\right)$, but not on $\partial\phi$ and
higher order derivatives. 
\item $X_{M}\left(\phi\right)$ is the Lie derivative of $\phi$ along $X_{M}$. 
\item $X_{Q}$ is the characteristic vector field, which corresponds to
the effective transformation of the fields from equation \ref{eq:trafo of the fields},
given by 
\begin{equation}
X_{Q}=\partial_{\epsilon}\vert_{0}\,Q_{\epsilon}\circ\phi\circ C_{\epsilon}^{-1}=X_{V}-X_{M}\left(\phi\right)
\end{equation}
\end{itemize}
It is also important to spend few words on the different types of
derivatives that are used in this geometrical construction. 
\begin{itemize}
\item The derivative $D_{V}$ denotes a derivative of the Lagrangian along
the coordinates of the fiber. In the common notation we can write
$\partial_{\phi}$ or $\delta_{\phi}$ 
\item $D_{J}$ denotes the derivative of the Lagrangian with respect to
the jet coordinates. In the above notation we can write $\partial_{\partial_{i}\phi}$
or $\delta_{\partial_{i}\phi}$. 
\item The derivative $D$ refers to the total derivative with respect to
the base manifold. This means that the implicit dependence on the
base manifold through fields needs to be taken into account. 
\item The partial derivative $\partial_{i}$ is instead a derivative purely
on the explicit dependence of the coordinates. Using above notation
we can write 
\begin{equation}
D\,f\left(x,\phi\left(x\right)\right)=\partial_{x}f+D_{V}\,f\cdot\partial_{x}\phi.
\end{equation}
We also use the capital letter in $\text{Div}\left(L\cdot X_{M}\right)$
for the total derivatives used in the divergence and $\text{div}\left(X_{M}\right)$
to denote the divergence taken only with respect to the explicit coordinates. 
\end{itemize}
Equation \ref{eq: Local symmetry condition} holds for local as well
as non-local Lagrangians. By partial integration equation \ref{eq: Local symmetry condition}
becomes 
\begin{equation}
E_{L}\left[X_{Q}\right]+\text{Div}\left(D_{J}L\cdot X_{Q}+L\cdot X_{M}\right)=0.
\end{equation}
Where $E_{L}$ is the Euler operator acting on the Lagrangian $L$\footnote{For local theories $E_{L}$ coincides with the equations of motion
and the above equation becomes the usual Noether identity. For non-local
theories, however, $E_{L}$ does not coincide with the equations of
motion, due to their integro-differential structure. In this case
further work needs to be done to provide a connection between the
equations of motion and the divergence terms. The resulting relation
is shown in the next section of this paper and is carefully derived
in \citep{Kegeles:2015oua}.}.

Having clarified the terminology and the notation, we can use \ref{eq: Local symmetry condition}
to derive the most general geometric symmetries of the various GFT
models.

\noindent We will use a rather standard procedure, based on the following
steps:

i) We assume a most general vector field on the vector bundle and
insert it in \ref{eq: Local symmetry condition}, ii) we rearrange
the resulting equation by different powers in derivatives of fields.
Since the coefficients $X_{M}^{i}$ and $X_{V}^{i}$ do not depend
on derivatives of the fields, it is possible to extract all powers
explicitly, iii) different powers of derivatives of $\phi$ are linearly
independent since the condition \ref{eq: Local symmetry condition}
has to be satisfied for all fields. For this reason the coefficients
in front of each term have to vanish separately. This results in simple
differential equations for the coefficients of the vector field which
can then be easily solved.

Since the GFT models of interest, here, are defined on many copies
of $SU\left(2\right)$ the notation can quickly become unreadable.
For this reason we summarize the notation used in the rest of this
section in the table \ref{tab:Usage-of-indices}.

\begin{table}[h]
\renewcommand*{\arraystretch}{2}%
\begin{tabular}{ccl}
\multirow{2}{*}{$u^{c}$  } &  & Field value $\phi^{c}\left(\vec{g}\right)$\tabularnewline
\cline{2-3} 
 & $c$  & - color of the field,\tabularnewline
\hline 
\multirow{4}{*}{$u_{iA}^{c}$ } &  & Derivative of the field $\phi^{c}$ at the point $\vec{g}$\tabularnewline
\cline{2-3} 
 & $i$  & - chart component of the single copy of $SU\left(2\right)$\tabularnewline
\cline{2-3} 
 & $A$  & - number of the copy of $SU\left(2\right)$\tabularnewline
\cline{2-3} 
 & $iA$  & direction of the derivative $\partial_{iA}\phi\vert_{\vec{g}}$ \tabularnewline
\hline 
\multirow{3}{*}{$X_{M}^{iA}$} &  & Vector field that acts on the base manifold $M$\tabularnewline
\cline{2-3} 
 & $i$  & - chart component of the single copy of $SU\left(2\right)$\tabularnewline
\cline{2-3} 
 & $A$  & - number of the copy of $SU\left(2\right)$\tabularnewline
\hline 
\multirow{2}{*}{$X_{u^{c}}$} &  & Component of $X_{V}$ in $\partial_{u^{c}}$ direction\tabularnewline
\cline{2-3} 
 & $c$  & - color of the transformed field\tabularnewline
\hline 
\end{tabular}

\caption{Usage of indices in this section\label{tab:Usage-of-indices}}
\end{table}

We denote the vector fields by $X_{M}$ and $X_{V}$ and refer to
their coefficients in a specific chart by $X_{M}^{iA}$ and $X_{u^{c}}$
respectively, i.e. 
\begin{align}
X_{M} & =X_{M}^{iA}\,\partial_{iA} & X_{V} & =X_{u^{c}}\,\partial_{u^{c}}+X_{\bar{u}^{c}}\,\partial_{\bar{u}^{c}}.
\end{align}
For the rest of this section we assume the summation convention over
repeated indices.

The local part of the action is given by 
\begin{equation}
L\left(u^{c},u_{J}^{c}\right)=\sum_{iA}\kappa\bar{u}_{iA}^{c}u_{iA}^{c}+m\,\bar{u}^{c}u^{c},\label{eq:local lagrangian}
\end{equation}
where the sum over $A$ ranges in $\left\{ 1,\cdots,n\right\} $ (the
$SU\left(2\right)$ copies of the local base manifold) 
\begin{widetext}
The above symmetry condition equation \ref{eq: Local symmetry condition}
implies 
\begin{equation}
X_{M}\left(L\right)+L\,\text{Div}\left(X_{M}\right)+2\kappa\,\partial_{iA}\phi^{c}\cdot D^{iA}\left(X_{cQ}\right)+2m\phi^{c}\,\left(X_{cQ}\right)=0.
\end{equation}
Explicitly sorting the terms by powers of $u_{iA}^{c}$ we get

\begin{align}
0= & \left[m\left|u^{c}\right|^{2}\text{div}\left(X_{M}\right)+m\bar{u}^{c}X_{u^{c}}+mu^{c}X_{\bar{u}^{c}}\right]\label{eq:1-1}\\
+\Re\left[u_{nA}^{t}\right] & \left[m\left|u^{c}\right|^{2}\left(D_{\bar{u}^{t}}X_{M}^{nA}+D_{u^{t}}X_{M}^{nA}\right)+\kappa g_{A}^{nm}\left(\partial_{mA}X_{u^{t}}+\partial_{mA}X_{\bar{u}^{t}}\right)\right]\label{eq:2-1}\\
+\imath\Im\left[u_{nA}^{t}\right] & \left[m\left|u^{c}\right|^{2}\left(D_{u^{t}}X_{M}^{nA}-D_{\bar{u}^{t}}X_{M}^{nA}\right)+\kappa g_{A}^{nm}\left(\partial_{mA}X_{\bar{u}^{t}}-\partial_{mA}X_{u^{t}}\right)\right]\label{eq:3-1}\\
+\kappa\Re\left[\bar{u}_{nA}^{c}u_{mA}^{c}\right] & \left[\frac{1}{2}X_{M}^{iB}\partial_{iB}g_{A}^{nm}-2g_{A}^{ni}\,\partial_{iA}X_{M}^{mA}+g_{A}^{nm}\left\{ \left(D_{u^{c}}X_{u^{c}}+D_{\bar{u}^{c}}X_{\bar{u}^{c}}\right)+\text{div}\left(X_{M}\right)\right\} \right]\label{eq:4}\\
-2\kappa\Re\left[\bar{u}_{nA}^{c}u_{mB\neq A}^{c}\right] & \left[g_{A}^{ni}\,\partial_{iA}X_{M}^{mB\neq A}\right]\label{eq:5}\\
+\kappa g_{A}^{nm}\Re\left[\bar{u}_{nA}^{c}u_{mA}^{t\neq c}\right] & \left[D_{u^{t\neq c}}X_{cV}+D_{\bar{u}^{c}}X_{\bar{u}^{t\neq cV}}\right]\label{eq:6}\\
+\imath\kappa g_{A}^{nm}\,\Im\left[\bar{u}_{nA}^{c}u_{mA}^{t\neq c}\right] & \left[D_{u^{t\neq c}}X_{u^{c}}-D_{\bar{u}^{c}}X_{\bar{u}^{t\neq c}}\right]\label{eq:6.7}\\
+\kappa g_{A}^{nm}\Re\left[\bar{u}_{nA}^{c}\bar{u}_{mA}^{t}\right] & \left[D_{\bar{u}^{t}}X_{c}+D_{u^{t}}X_{\bar{u}^{c}}\right]\label{eq:7}\\
+\imath\kappa g_{A}^{nm}\,\Im\left[\bar{u}_{nA}^{c}\bar{u}_{mA}^{t}\right] & \left[D_{\bar{u}^{t}}X_{u^{c}}-D_{u^{t}}X_{\bar{u}^{c}}^{\dagger}\right]\label{eq:8}\\
+2\kappa\bar{u}_{nA}^{c}u_{iA}^{c}\,\Re\left[u_{mA}^{t}\right] & \left[-2g_{A}^{nm}\left(D_{u^{t}}X_{M}^{iA}+D_{\bar{u}^{t}}X_{M}^{iA}\right)+g_{A}^{ni}\left(D_{u^{t}}X_{M}^{mA}+D_{\bar{u}^{t}}X_{M}^{mA}\right)\right]\label{eq:9}\\
+\imath2\kappa\bar{u}_{nA}^{c}u_{iA}^{c}\,\Im\left[\bar{u}_{mA}^{t}\right] & \left[-2g_{A}^{nm}\left(D_{u^{t}}X_{M}^{iA}-D_{\bar{u}^{t}}X_{M}^{iA}\right)+g_{A}^{ni}\left(D_{u^{t}}X_{M}^{mA}-D_{\bar{u}^{t}}X_{M}^{mA}\right)\right]\label{eq:10}\\
+\bar{u}_{nA}^{c}u_{iB\neq A}^{c}\,\Re\left[u_{mA}^{t}\right] & \left[-2\kappa g_{A}^{nm}\left(D_{u^{t}}X_{M}^{iB\neq A}+D_{\bar{u}^{t}}X_{M}^{iB\neq A}\right)\right]\label{eq:11}\\
+\imath\bar{u}_{nA}^{c}u_{iB\neq A}^{c}\,\Im\left[\bar{u}_{mA}^{t}\right] & \left[-2\kappa g_{A}^{nm}\left(D_{u^{t}}X_{M}^{iB\neq A}-D_{\bar{u}^{t}}^{\dagger}X_{M}^{iB\neq A}\right)\right].\label{eq:12}
\end{align}
\end{widetext}

This equation has to hold true for arbitrary fields $u^{c}$ and $u_{iA}^{c}$.
However, the parts in brackets do not depend on $u_{iA}^{c}$, which
implies that each line has to vanish individually\footnote{ Notice that, if we allowed for derivative dependence of the coefficients
$\chi=\chi\left(x,\phi\left(x\right),D\phi\vert_{x}\right)$ and similar
for the $\xi$, we could not argue that the terms with different powers
of derivatives of $\phi$ have to vanish independently, since the
terms in brackets would also contain derivatives of the fields.}. The consequences of these equations read: 
\begin{enumerate}
\item Equations \ref{eq:6} and \ref{eq:6.7} imply that the vector field
component $X_{u^{c}}$ depend only on the field colors they transform,
that is (no summation) 
\begin{align*}
X_{u^{c}} & =X_{u^{c}}\left(\vec{g},u^{c},\bar{u}^{c}\right) & X_{\bar{u}^{c}} & =X_{\bar{u}^{c}}\left(\vec{g},u^{c},\bar{u}^{c}\right).
\end{align*}
\item Equations \ref{eq:7} and \ref{eq:8} additionally imply that the
vector fields $X_{u^{c}}$ do not depend on the complex conjugate
of the field, that is 
\begin{align*}
X_{u^{c}} & =X_{u^{c}}\left(\vec{g},u^{c}\right) & X_{\bar{u}^{c}} & =X_{\bar{u}^{c}}\left(\vec{g},\bar{u}^{c}\right).
\end{align*}
\item Equations \ref{eq:11} and \ref{eq:12} tell us that the vector fields
that transform the base manifold do not depend on the field values
$u^{c}$ i.e. $X_{M}^{A}=X_{M}^{A}\left(\vec{g}\right)$. From this
condition, equations \ref{eq:9} and \ref{eq:10} are automatically
satisfied. 
\item Due to the above, equations \ref{eq:2-1} and \ref{eq:3-1} reduce
to 
\begin{equation}
\partial_{mA}X_{u^{t}}=0=\partial_{mA}X_{\bar{u}^{c}}.
\end{equation}
That is, the vector fields do not explicitly depend on the points
in the base manifold 
\begin{align*}
X_{cV} & =X_{u^{c}}\left(u^{c}\right) & X_{\bar{u}^{c}} & =X_{\bar{u}^{c}}\left(\bar{u}^{c}\right).
\end{align*}
\item Equation \ref{eq:1-1}, together with the above conclusion, restricts
the vector fields to a specific form 
\begin{align}
X_{u^{c}} & =C\,u^{c} & X_{\bar{u}^{c}} & =\bar{C}\,\bar{u}^{c},
\end{align}
where $C$ is an arbitrary constant that satisfies 
\begin{equation}
\text{div}\left(X_{M}\right)=-C-\bar{C}.\label{eq:u1 symmetry condition}
\end{equation}
\item The above condition reduces equations \ref{eq:4} and \ref{eq:5}
to 
\begin{eqnarray*}
X_{M}^{iB}\partial_{iB}g_{A}^{nm}-2\,g_{A}^{ni}\,\partial_{iA}X_{M}^{mA}-2\,g_{A}^{mi}\,\partial_{iA}X_{M}^{nA} & = & 0\\
g_{A}^{ni}\,\partial_{iA}X_{M}^{mB\neq A}+g_{B}^{ni}\,\partial_{iA}X_{M}^{mA\neq B} & = & 0.
\end{eqnarray*}
\end{enumerate}
These two equations are the only ones that are not trivial to solve.
However, although lengthy, their solution can be found in a straightforward
way. The solution in Hopf coordinates $\left(\eta,\xi,\chi\right)$\footnote{In this coordinates the metric on $SU\left(2\right)$ is given by
$g=\text{d}\eta^{2}+\sin^{2}\eta\,\text{d}\xi^{2}+\cos^{2}\eta\,\text{d}\chi^{2}$.} reads as

\begin{eqnarray}
X_{M}^{\eta A} & = & C_{1}\sin\xi_{A}\sin\chi_{A}\label{eq:sVF1}\\
 &  & +C_{2}\cos\xi_{A}\sin\chi_{A}\nonumber \\
 &  & +C_{3}\sin\xi_{A}\cos\chi_{A}\nonumber \\
 &  & +C_{4}\cos\xi_{A}\cos\chi_{A}\nonumber \\
X_{M}^{\xi A} & = & \frac{\cos\eta_{A}}{\sin\eta_{A}}\partial_{\xi A}X_{M}^{\eta A}+C_{5}\label{eq:sVF2}\\
X_{M}^{\chi A} & = & -\frac{\sin\eta_{A}}{\cos\eta_{A}}\partial_{\chi A}X_{M}^{\eta A}+C_{6},\label{eq:sVF3}
\end{eqnarray}
where $C_{i}$'s are arbitrary constants.

Setting subsequently $C_{i}$ to one and the rest of the coefficients
to zero we obtain, for each copy of the group $A$, six linearly independent
vector fields given by 
\begin{eqnarray}
v_{1} & = & \left(\begin{array}{c}
\sin\left(\xi\right)\sin\left(\chi\right)\\
\cot\left(\eta\right)\sin\left(\xi\right)\cos\left(\chi\right)\\
-\tan\left(\eta\right)\sin\left(\xi\right)\cos\left(\chi\right)
\end{array}\right)\\
v_{2} & = & \left(\begin{array}{c}
\cos\left(\xi\right)\sin\left(\chi\right)\\
-\cot\left(\eta\right)\sin\left(\xi\right)\sin\left(\chi\right)\\
-\tan\left(\eta\right)\cos\left(\xi\right)\cos\left(\chi\right)
\end{array}\right)\\
v_{3} & = & \left(\begin{array}{c}
\sin\left(\xi\right)\cos\left(\chi\right)\\
\cot\left(\eta\right)\cos\left(\xi\right)\cos\left(\chi\right)\\
\tan\left(\eta\right)\sin\left(\xi\right)\sin\left(\chi\right)
\end{array}\right)\\
v_{4} & = & \left(\begin{array}{c}
\cos\left(\xi\right)\cos\left(\chi\right)\\
-\cot\left(\eta\right)\sin\left(\xi\right)\cos\left(\chi\right)\\
\tan\left(\eta\right)\cos\left(\xi\right)\sin\left(\chi\right)
\end{array}\right),
\end{eqnarray}
and 
\begin{align}
v_{5} & =\left(\begin{array}{c}
0\\
1\\
0
\end{array}\right) & v_{6} & =\left(\begin{array}{c}
0\\
0\\
1
\end{array}\right).
\end{align}
It is a direct calculation to check that these vector fields are divergence
free, $\text{div}\left(V_{i}\right)=0$. This fact, together with
equation \ref{eq:u1 symmetry condition}, implies 
\begin{align}
X_{u^{k}} & =\imath C_{k}\,u^{k} & X_{\bar{u}^{k}} & =-\imath C_{k}\,\bar{u}^{k}\quad,
\end{align}
which generates the usual $U\left(1\right)$ symmetry of fields for
each color.

In order to find the symmetry group generated by the fields $v_{1},\cdots,v_{6}$
, we look at their algebra. The six dimensional Lie algebra of $v_{1},\cdots,v_{6}$
is given in table \ref{tab:Lie-algebra-of}

\begin{table}
\renewcommand*{\arraystretch}{2}%
\begin{tabular}{>{\raggedleft}p{1cm}|>{\raggedleft}p{1cm}>{\raggedleft}p{1cm}>{\raggedleft}p{1cm}>{\raggedleft}p{1cm}>{\raggedleft}p{1cm}>{\raggedleft}p{1cm}}
 & $v_{1}$  & $v_{2}$  & $v_{3}$  & $v_{4}$  & $v_{5}$  & $v_{6}$\tabularnewline
\hline 
$v_{1}$  & 0  & $v_{5}$  & $v_{6}$  & 0  & $-v_{2}$  & $-v_{3}$\tabularnewline
$v_{2}$  & $-v_{5}$  & 0  & 0  & $v_{6}$  & $v_{1}$  & $-v_{4}$\tabularnewline
$v_{3}$  & $-v_{6}$  & 0  & 0  & $v_{5}$  & $-v_{4}$  & $v_{1}$\tabularnewline
$v_{4}$  & 0  & $-v_{6}$  & $-v_{5}$  & 0  & $v_{3}$  & $v_{2}$\tabularnewline
$v_{5}$  & $v_{2}$  & $-v_{1}$  & $v_{4}$  & $-v_{3}$  & 0  & 0\tabularnewline
$v_{6}$  & $v_{3}$  & $v_{4}$  & $-v_{1}$  & $-v_{2}$  & 0  & 0\tabularnewline
\end{tabular}

\caption{\label{tab:Lie-algebra-of}Lie algebra of symmetry vector fields}
\end{table}

We can split this algebra into $\mathfrak{su}\left(2\right)\times\mathfrak{su}\left(2\right)$
by taking the following linear combinations 
\begin{align}
l_{1} & =\frac{v_{5}-v_{6}}{2} & r_{1} & =\frac{v_{5}+v_{6}}{2}\nonumber \\
l_{2} & =\frac{v_{3}-v_{2}}{\sqrt{2}} & r_{2} & =\frac{v_{3}+v_{2}}{\sqrt{2}}\label{eq:left}\\
l_{3} & =\frac{v_{4}+v_{1}}{\sqrt{2}} & r_{3} & =\frac{v_{4}-v_{1}}{\sqrt{2}}\quad.\nonumber 
\end{align}
The commutators for $l_{i}$ and $r_{i}$ become 
\begin{align}
\left[l_{1},l_{2}\right] & =l_{3} & \left[r_{1},r_{2}\right] & =r_{3}\\
\left[l_{1},l_{3}\right] & =-l_{2} & \left[r_{1},r_{3}\right] & =-r_{2}\\
\left[l_{2},l_{3}\right] & =2l_{1} & \left[r_{2},r_{3}\right] & =2r_{1}
\end{align}
\begin{equation}
\left[l_{i},r_{j}\right]=0\quad\forall i,j\in\left\{ 1,2,3\right\} 
\end{equation}
A closer inspection shows that $l_{i}$ and $r_{i}$ form a set of
left and right invariant vector fields on $SU\left(2\right)$, respectively
\citep{Akhtarshenas2010aaa}.

Since the above algebra was derived for each copy of the group $A$,
the whole symmetry group of the local part of the action becomes 
\begin{equation}
\left[SU\left(2\right)\times SU\left(2\right)\right]^{\times3}\times U\left(1\right)^{\times Nc},
\end{equation}
acting on the base manifold by left and right multiplication as 
\begin{equation}
L_{\vec{\eta}}\circ R_{\vec{\mu}}\left(\vec{g}\right)=\left(\vphantom{\vec{L}}\eta_{1}g_{1}\mu_{1},\,\eta_{2}g_{2}\mu_{2},\,\eta_{3}g_{3}\mu_{3}\right),
\end{equation}
for any $\vec{\eta},\vec{\mu}\in SU\left(2\right)^{\times3}$ and
on fields by multiplication with a $U\left(1\right)$ phase.

It is now trivial to insert this transformations in the interaction
part of the action in order to verify which of the transformations
remains a symmetry. It is easy to see that the symmetry group is preserved
for tensorial interactions.

Indeed, this is a remarkable feature of tensorial GFTs, which can
be also stated as follows: by their very definition, the symmetry
group of tensorial interactions is the same as that of the local part
of the action. In this sense, this is a confirmation of the very motivation
for introducing tensorial interactions as encoding the correct new
notion of \textsl{locality} for tensorial field theories \citep{Rivasseau:2011hm}.

This is due to the fact that both of the symmetry groups we have found
above form particular cases of the unitary transformations characterizing
tensorial interactions, as the $U\left(1\right)$ transformations
are implemented by 
\begin{equation}
U\left(g,h\right)=\delta\left(gh^{-1}\right)e^{\imath\theta},
\end{equation}
and the left (right) multiplication by the group is obtained for fixed
$\eta$ or $\mu$ as 
\begin{align}
U\left(g,h\right) & =\delta\left(\eta gh^{-1}\right) & U\left(g,h\right) & =\delta\left(g\mu h^{-1}\right).
\end{align}

In the case of simplicial models the status of both the $SU\left(2\right)$
and the $U\left(1\right)$ group as symmetries of the full theory
depend on the specific interaction in question. We need, then, to
check explicitly the condition on the group action \ref{eq:condition on the group action}.
We postpone the verification of this condition to the next section,
since it will be the main tool in the analysis of gauge invariant
models.

\subsection{Gauge invariant models}

In this section we study more in detail the symmetry group of simplicial
GFT interactions, and we show how the treatment can be significantly
simplified in the presence of gauge invariance. Contrary to the previous
case, we will use the interaction part to classify the symmetry group,
subsequently checking which of the symmetries represent also a symmetry
of the local action. The first part of the treatment is independent
of the local part of the action and holds for a large number of base
group manifolds, which is why we do not specify the group at the beginning
of the section. However, in order to verify the symmetry group for
the local part we need to know the exact structure of the differential
operator involved and so we need to specify the underlying group as
well. From this point onwards, we specialize the notation to the $n=3$
case for simplicity of exposition. The extension to generic $n$ is
straightforward.

\subsubsection{Admissible base manifold transformations}

The combinatorics of the interaction part is encoded in the function
$f$ from equations \ref{eq:trafo of the sections} and \ref{eq:Relation between C's}.
For example the combinatorial structure of a 3d simplicial interaction
is given by 
\begin{gather*}
f:\left(g_{1},\cdots,g_{6}\right)\\
\mapsto\left(g_{1},g_{2},g_{3}\right)\left(g_{3},g_{4},g_{5}\right)\left(g_{5},g_{2},g_{6}\right)\left(g_{6},g_{4},g_{1}\right).
\end{gather*}
Admissible transformations of the base manifold are given by those
functions $C:G^{\times3}\to G^{\times3}$ that satisfy the relation
\ref{eq:Relation between C's}. Therefore, for any $\left(g_{1},\cdots,g_{6}\right)\in G^{\times6}$,
there should exist a point $\left(\tilde{g}_{1},\cdots\tilde{g}_{6}\right)\in G^{\times6}$
such that 
\begin{equation}
C^{\times4}\circ f\left(g_{1},\cdots,g_{6}\right)=f\left(\tilde{g}_{1},\cdots,\tilde{g}_{6}\right).\label{eq:Condition for Boulatov}
\end{equation}
Writing $C$ in components as 
\begin{equation}
C\left(\vec{g}\right)=\left(C^{1}\left(\vec{g}\right),C^{2}\left(\vec{g}\right),C^{1}\left(\vec{g}\right)\right),
\end{equation}
condition \ref{eq:Condition for Boulatov} implies 
\begin{align}
C^{1}\left(g_{1},g_{2},g_{3}\right) & =C^{3}\left(g_{6},g_{4},g_{1}\right)\\
C^{2}\left(g_{1},g_{2},g_{3}\right) & =C^{2}\left(g_{5},g_{2},g_{6}\right),
\end{align}
which suggests the following decomposition of $C$, 
\begin{equation}
C\left(g_{1},g_{2},g_{3}\right)=C^{1}\left(g_{1}\right)C^{2}\left(g_{2}\right)C^{3}\left(g_{3}\right).
\end{equation}
Notice that in this case the diffeomorphism properties of $C$ carry
over to the components $C^{i}$.

According to equation \ref{eq:trafo of the sections}, the fields
transform under $C$ as 
\begin{equation}
\phi\mapsto\tilde{\phi}=\phi\circ C^{-1}.
\end{equation}
The field $\tilde{\phi}$ needs to be gauge invariant as well, otherwise
the transformation $C$ would leave the allowed space of fields. The
gauge invariance of $\tilde{\phi}$ reads 
\begin{equation}
\tilde{\phi}\circ R_{h}=\phi\circ C\circ R_{h}\overset{!}{=}\phi\circ C=\tilde{\phi}.
\end{equation}
Since this has to be true for all gauge invariant fields $\phi$,
the point $C\circ R_{h}\left(\vec{g}\right)$ needs to be in the same
orbit (under the multiplication from the right by the diagonal group)
as the point $C\left(\vec{g}\right)$. This means that, for any $h\in G_{D}$,
there should exist an $\tilde{h}\in G_{D}$ such that 
\begin{equation}
C\circ R_{h}=R_{\tilde{h}}\circ C,
\end{equation}
or point-wise 
\begin{equation}
C\left(\vec{g}h\right)=C\left(\vec{g}\right)\tilde{h}.\label{eq:gauge invariance condition}
\end{equation}
As we show in the appendix \ref{sec:Reduction-of-transformations},
this restricts the $C$, up to discrete transformations, to be of
the form 
\begin{equation}
C\left(\vec{g}\right)=\vec{L}\cdot h^{-1}\,\vec{g}\,h,
\end{equation}
for some $L\in G^{\times2}$ and $h\in G_{D}$.

In the end, the symmetry group of the interaction part becomes 
\begin{equation}
G^{\times2}\times G_{D}.
\end{equation}
It is evident that this group already forms a symmetry group, due
to the left invariance of the Haar measure. We can summarize the role
of combinatorial structure and the gauge invariance on the transformation
group of the base manifolds as follows 
\begin{widetext}
\begin{equation}
\text{Diff}\left(\left[G^{^{\times3}}\right]\right)\underset{\text{combinatorics}}{\longrightarrow}\text{Diff}\left(G\right)^{\times2}\underset{\text{gauge invariance}}{\longrightarrow}G^{\times2}\times G_{3D}.
\end{equation}
\end{widetext}

\begin{table*}[t]
\begin{centering}
\renewcommand*{\arraystretch}{2}%
\begin{tabular}{|>{\centering}m{6cm}|>{\centering}m{4.5cm}|>{\centering}m{6.5cm}|}
\hline 
\textbf{Model}  & \textbf{Symmetry Group}  & \textbf{Action}\tabularnewline
\hline 
\hline 
$\phi_{1,2,3}^{1}\phi_{1,4,5}^{2}\phi_{6,2,5}^{3}\phi_{6,4,3}^{4}$  & $G^{\times3}\times U\left(1\right)^{\times3}$  & $\begin{array}{ll}
\vec{g} & \mapsto L_{C}\left(\vec{g}\right)\\
\phi^{c} & \mapsto e^{\imath\theta_{c}}\phi^{c}\quad\sum_{c}\theta^{c}=0
\end{array}$

with $C=\left(c^{1},c^{2},c^{3}\right)$\tabularnewline
\hline 
$\phi_{1,2,3}\,\phi_{3,4,5}\,\phi_{5,2,6}\,\phi_{6,4,1}$  & $G^{\times2}$  & $\begin{array}{ll}
\vec{g} & \mapsto L_{C}\left(\vec{g}\right)\end{array}$

with $C=\left(c^{1},c^{2},c^{1}\right)$\tabularnewline
\hline 
$\phi_{1,2,3}^{P}\,\phi_{1,4,5}^{P}\,\phi_{2,5,6}^{P}\,\phi_{3,6,4}^{P}$  & $G$  & $\begin{array}{ll}
\vec{g} & \mapsto L_{C}\left(\vec{g}\right)\end{array}$

with $C=\left(c,c,c\right)$\tabularnewline
\hline 
$\phi_{1,2,3}\,\bar{\phi}_{3,4,5}\,\phi_{5,2,6}\,\bar{\phi}_{6,4,1}$  & $G^{\times2}\times U\left(1\right)$  & $\begin{array}{ll}
\vec{g} & \mapsto L_{C}\left(\vec{g}\right)\\
\phi & \mapsto e^{\imath\theta}\phi
\end{array}$

with $C=\left(c^{1},c^{2},c^{1}\right)$\tabularnewline
\hline 
$\phi_{1,2,3,4}\,\phi_{4,5,6,7}\,\phi_{7,3,8,9}\,\phi_{9,6,2,10}\,\phi_{10,8,5,1}$  & $SU\left(2\right)^{\times2}$  & $\begin{array}{ll}
\vec{g} & \mapsto L_{C}\left(\vec{g}\right)\end{array}$

with $C=\left(c^{1},c^{2},c^{2},c^{1}\right)$\tabularnewline
\hline 
Barrett-Crane  & $Spin\left(4\right)^{\times2}\times SU\left(2\right)$  & $\begin{array}{ll}
\vec{g} & \mapsto L_{S}\left(\vec{g}\right)\\
\vec{g} & \mapsto c\cdot\left(\vec{g};k\right)\cdot c^{-1}
\end{array}$

with $S=\left(s^{1},s^{2},s^{2},s^{1}\right)$\tabularnewline
\hline 
\end{tabular}
\par\end{centering}
\centering{}\caption{Models and their symmetry groups excluding gauge symmetries of the
fields.\label{tab:Models-and-their symmetries} Hereby, $c^{i}\in SU\left(2\right)$
and $s^{i}\in Spin\left(4\right)$ and $L_{C}$ denotes the left multiplication
by $C$.}
\end{table*}
For higher dimensional models such as the Ooguri model with the interaction
given by 
\begin{gather*}
f:\left(g_{1},\cdots,g_{10}\right)\mapsto\\
\left(g_{1},g_{2},g_{3},g_{4}\right)\left(g_{4},g_{5},g_{6},g_{7}\right)\left(g_{7},g_{3},g_{8},g_{9}\right)\times\\
\times\left(g_{9},g_{6},g_{2},g_{10}\right)\left(g_{10},g_{8},g_{5},g_{1}\right),
\end{gather*}
we observe that the above treatment still results in the transformation
group 
\begin{equation}
G^{2}\times G_{4D},\label{eq:Ooguri model}
\end{equation}
where $G^{2}$ acts by left multiplication as 
\begin{equation}
\left(G_{1},G_{2},G_{2},G_{1}\right)\vec{g}=\left(G_{1}g^{1},G_{2}g^{2},G_{2}g^{3},G_{1}g^{4}\right).
\end{equation}

\subsubsection*{Note on differences between simplicial combinatorics}

As we mentioned above, the combinatorial structure for simplicial
models can vary. This variation is captured by different functions
$f$, as used in the beginning of this section.

For the original Boulatov interaction we get 
\begin{gather*}
f:\left(g_{1},\cdots,g_{6}\right)\\
\mapsto\left(g_{1},g_{2},g_{3}\right)\left(g_{1},g_{4},g_{5}\right)\left(g_{2},g_{5},g_{6}\right)\left(g_{3},g_{6},g_{4}\right).
\end{gather*}
The resulting transformations become 
\begin{equation}
C\left(\vec{g}\right)=C^{1}\left(g_{1}\right)C^{1}\left(g_{2}\right)C^{1}\left(g_{3}\right),
\end{equation}
where all the components are the same. It is easy to check that this
transformation also respect the cyclic permutation condition and therefore
we get the symmetry group of the Boulatov model as 
\begin{equation}
G\times G.
\end{equation}

In the colored case, instead, we get 
\begin{gather*}
f:\left(g_{1},\cdots,g_{6}\right)\\
\mapsto\left(g_{1},g_{2},g_{3}\right)\left(g_{1},g_{4},g_{5}\right)\left(g_{6},g_{2},g_{5}\right)\left(g_{6},g_{4},g_{3}\right),
\end{gather*}
and the resulting admissible transformations are 
\begin{equation}
C\left(\vec{g}\right)=C^{1}\left(g_{1}\right)C^{2}\left(g_{2}\right)C^{3}\left(g_{3}\right).
\end{equation}
The group of admissible transformations is therefore 
\begin{equation}
G^{\times3}\times G.
\end{equation}

\subsubsection{Symmetries of gauge invariant models}

Following the procedure of the previous section, the infinitesimal
symmetry condition for the simplicial interaction in three dimensions
takes the form 
\begin{equation}
0=D_{V}L\cdot X_{Q}+\text{Div}\left(L\,X_{M}\right),
\end{equation}
Writing the same condition in terms of $u^{i}$, and using the fact
that the only admissible base manifold transformations are generated
by divergence free vector fields, we obtain 
\begin{align}
u^{2}u^{3}u^{4}\,X_{V}^{1}+u^{1}u^{3}u^{4}\,X_{V}^{2}\nonumber \\
+u^{1}u^{2}u^{4}\,X_{V}^{3}+u^{1}u^{2}u^{3}\,X_{V}^{4}\label{eq:Non-local term}\\
+c.c. & =0.\nonumber 
\end{align}
Hereby $X^{i}$ is evaluated at the point $\left(\vec{g}_{i},\vec{u}\right)$
with $\vec{g}_{1}=\left(g_{1},g_{2},g_{3}\right)$, $\vec{g}_{2}=\left(g_{3},g_{4},g_{5}\right)$,
$\vec{g}_{3}=\left(g_{5},g_{2},g_{6}\right)$ and $\vec{g}_{4}=\left(g_{6},g_{4},g_{1}\right)$
and $\vec{u}=\left\{ \phi^{i}\left(\vec{g}\right)\right\} _{i\in\left\{ 1,\cdots,4\right\} }$.
Notice that $\vec{u}$ is not $\left(u^{1},u^{2},u^{3},u^{4}\right)$
since the latter tuple is given by $\left\{ \phi^{i}\left(\vec{g}_{i}\right)\right\} _{i\in\left\{ 1,\cdots,4\right\} }$.
Equation \ref{eq:Non-local term} needs to hold true for any $\phi^{i}$
and any point of the base manifold.

Inserting the formal power series expansion of $X_{V}^{i}$ 
\[
X_{V}^{i}\left(\vec{g},\vec{u}\right)=\sum_{\vec{m}}\Theta_{\vec{m}}^{i}\left(\vec{g}\right)\phi^{1}\left(\vec{g}\right)^{m_{1}}\cdots\phi^{4}\left(\vec{g}\right)^{m_{4}}\bar{\phi}^{1}\left(\vec{g}\right)\cdots\bar{\phi}^{4}\left(\vec{g}\right),
\]
in equation \ref{eq:Non-local term} we observe that all the coefficient
functions $\Theta^{i}$ vanish except for one, such that 
\begin{equation}
X_{V}^{i}\left(\vec{g},\vec{u}\right)=\Theta^{i}\left(\vec{g}\right)u^{i}.
\end{equation}
Equation \ref{eq:Non-local term} becomes 
\begin{equation}
\Theta^{1}\left(\vec{g}_{1}\right)+\Theta^{2}\left(\vec{g}_{2}\right)+\Theta^{3}\left(\vec{g}_{3}\right)+\Theta^{4}\left(\vec{g}_{4}\right)=0.
\end{equation}
As we show in the appendix \ref{sec:Constancy-of-the}, the only functions
that are gauge invariant and satisfy the above equation are constants,
\begin{align}
\theta^{i} & =const. & \sum_{i}\theta^{i} & =0\quad.\label{eq:condition on the phase}
\end{align}
Therefore $X_{V}^{i}$ generate the symmetry group $U\left(1\right)^{\#c-1}$,
where $\#c$ is the number of colors in the interaction part of the
model. Notice, that if the model is not colored, that is the number
of colors is one, the $U\left(1\right)$ symmetry is not present.
The overall symmetry group for simplicial models becomes 
\begin{equation}
G^{\times n}\times G\times U\left(1\right)^{\#c-1},
\end{equation}
where $n$ depends on the actual combinatorial pattern, as we have
shown. This classification of symmetries also fits the model with
the interaction part of the type $\int\phi\bar{\phi}\phi\bar{\phi}$,
since we defined it as a model with two different colors.

The symmetry group of colored models, which is independent from the
precise combinatorial pattern of field arguments, is the largest compatible
with the one of the local part of the action (and with gauge invariance),
and coincides with the one of the corresponding tensorial model. This
gives a different perspective, and confirms, the close relation between
colored simplicial models and tensorial ones, highlighted first in
\citep{Bonzom:2012hw} in terms of properties of the corresponding
functional integrals.

\subsubsection{Barrett-Crane model}

We now briefly discuss the implication of the simplicity constraints,
in the Barrett-Crane formulation, on the symmetry group.

Applying the above analysis to the BC model from equation \ref{eq:BC interaction}
defined by the following function $f$ 
\begin{gather}
f:\left(g_{1},\cdots g_{10};k_{1},\cdots,k_{4}\right)\mapsto\\
\left(g_{1,2,3,4};k_{1}\right)\left(g_{4,5,6,7};k_{2}\right)\left(g_{7,3,8,9};k_{3}\right)\nonumber \\
\times\left(g_{9,6,2,10};k_{4}\right)\left(g_{10,8,5,1};k_{5}\right),
\end{gather}
we realize that the symmetry group for the gauge invariant BC model
without simplicity constrains would be that of an extended Ooguri
model from equation \ref{eq:Ooguri model} where the group $G$ is
now specified to $Spin\left(4\right)$

\begin{equation}
Spin\left(4\right)^{\times2}\times Spin\left(4\right)\times G\left(k\right).\label{eq:BC without Simplicity}
\end{equation}
The group $G\left(k\right)$ denotes a group of transformations of
the $SU\left(2\right)$ element $k_{i}$. However, remember that the
extension of the GFT field the $SU\left(2\right)$ variable $k$ was
needed for consistent implication of simplicity constrains and therefore
the actual meaning of $G\left(k\right)$ is relevant only after the
imposition of simplicity constrains. 

Equation \ref{eq:BC without Simplicity} provides the symmetry group
of extended Ooguri model with gauge invariance, in order to obtain
the symmetry group of the BC model simplicity constraints need to
be further imposed. We refer to the appendix \ref{sec:Barrett-Crane-model}
for explicit calculations and state here just the result of imposing
the simplicity constrains on the field $\phi$, by imposing invariance
under the projector $\mathcal{S}$ from equation \ref{eq:Simplicity constraint}.
As we show in the appendix \ref{sec:Barrett-Crane-model}, the simplicity
constrains 
\begin{equation}
\phi\circ\mathcal{S}=\phi,
\end{equation}
reduce the symmetry group of the Ooguri model (for the chosen combinatorics)
down to 
\begin{equation}
Spin\left(4\right)^{\times2}\times SU\left(2\right),
\end{equation}
where the $SU\left(2\right)$ group replaces the $Spin\left(4\right)\times G\left(k\right)$
part from equation \ref{eq:BC without Simplicity} and acts on the
elements of the local base manifold of the BC model $Spin\left(4\right)^{\times4}\times SU\left(2\right)$
by conjugation as 
\[
c\circ\left(\vec{g}_{-},\vec{g}_{+};k\right):=\left(c,c;c\right)\cdot\left(\vec{g}_{-},\vec{g}_{+};k\right)\cdot\left(c^{-1},c^{-1};c^{-1}\right).
\]
And $Spin(4)^{\times2}$ acts by the left multiplication 
\[
\left(G_{1},G_{2},G_{2},G_{1}\right)\left(\vec{g};k\right)=\left(G_{1}g^{1},G_{2}g^{2},G_{2}g^{3},G_{1}g^{4};k\right).
\]
The same considerations we have made regarding the dependance of the
symmetry group on the combinatorics and on the use of colors apply
also to the Barrett-Crane case. 

In table \ref{tab:Models-and-their symmetries} we summarize the symmetries
of different interaction terms.

\section{Classical currents}

\begin{table*}[t]
\begin{centering}
\renewcommand*{\arraystretch}{2}%
\begin{tabular}{|>{\centering}m{6cm}|>{\centering}m{4.5cm}|>{\centering}m{6.5cm}|}
\hline 
\textbf{Model}  & \textbf{Symmetry Group }  & \textbf{Action}\tabularnewline
\hline 
\hline 
\multirow{2}{6cm}{\centering$\phi_{1,2,3}^{1}\phi_{1,4,5}^{2}\phi_{6,2,5}^{3}\phi_{6,4,3}^{4}$} & $SU\left(2\right)^{\times3}$  & $\Delta=-2\lambda\sum_{c}\int\delta^{c}\,\Re\left[X_{cM}^{\oplus3}\left(\phi^{1}\phi^{2}\phi^{3}\phi^{4}\right)\right]$\tabularnewline
\cline{2-3} 
 & $U\left(1\right)^{\times3}$  & $\Delta=-\imath2\lambda\,\sum_{c}\theta_{c}\int\delta^{c}\,\Im\left(\phi^{1}\phi^{2}\phi^{3}\phi^{4}\right)$\tabularnewline
\hline 
$\phi_{1,2,3}\,\phi_{3,4,5}\,\phi_{5,2,6}\,\phi_{6,4,1}$  & $G^{\times2}$  & $\Delta=-8\lambda\,\Re\left(X_{M}^{\oplus2}\left(\phi\right)\int\phi\phi\phi\right)$\tabularnewline
\hline 
$\phi_{1,2,3}^{P}\,\phi_{4,3,5}^{P}\,\phi_{5,2,6}^{P}\,\phi_{6,4,1}^{P}$  & $G$  & $\Delta=-8\lambda\,\Re\left(X_{M}\left(\phi\right)\int\phi\phi\phi\right)$\tabularnewline
\hline 
\multirow{2}{6cm}{\centering$\phi_{1,2,3}\,\bar{\phi}_{3,4,5}\,\phi_{5,2,6}\,\bar{\phi}_{6,4,1}$ } & $G^{\times2}$  & $\Delta=-4\lambda\,\Re\left(X_{M}\left(\phi\right)\int\bar{\phi}\phi\bar{\phi}\right)$\tabularnewline
\cline{2-3} 
 & $U\left(1\right)$  & $\Delta=-4\lambda\theta\,\Im\left(\phi\int\bar{\phi}\phi\bar{\phi}\right)$\tabularnewline
\hline 
$\phi_{1,2,3,4}\,\phi_{4,5,6,7}\,\phi_{7,3,8,9}\,\phi_{9,6,2,10}\,\phi_{10,8,5,1}$  & $SU\left(2\right)^{\times2}$  & $\Delta=-10\lambda\,\Re\left(X_{M}^{\oplus2}\left(\phi\right)\int\phi\phi\phi\phi\right)$\tabularnewline
\hline 
\multirow{2}{6cm}{\centering Barrett-Crane} & $Spin\left(4\right)^{\times2}$  & $\Delta=-10\lambda\,\Re\left(X_{M}^{\oplus2}\left(\phi\right)\int\phi\phi\phi\phi\right)$\tabularnewline
\cline{2-3} 
 & $SU\left(2\right)$  & $\Delta=-10\lambda\,\Re\left(X_{M}\left(\phi\right)\int\phi\phi\phi\phi\right)$\tabularnewline
\hline 
\end{tabular}
\par\end{centering}
\centering\caption{Models and their correspondent correction terms. The vector fields
$X_{M}$ are the left invariant vector fields given in equation \ref{eq:left}
\label{tab:Models-and-their-correction terms}}
\end{table*}

We will now derive the (generalized) conservation laws for the symmetries
we identified in the last section. Once more we limit ourselves to
the classical regime of the GFTs, postponing the analysis of the full
quantum theory. Also, we stress again that the conservation laws and
corresponding currents, just like the whole kinematics and dynamics
of such quantum field theories, should not be interpreted in spatiotemporal
or geometric terms, at least in general. Even for GFT models with
a direct quantum gravity interpretation, the spatiotemporal and geometric
meaning of the various aspects and regimes of each model should be
extracted and analyzed with care. On this note, we point out that
the classical GFT equations of motion of 4d quantum gravity models,
which capture the hydrodynamics of special condensate states of the
theory, have been given a cosmological interpretation and have been
studied in some detail and with remarkable results in a series of
recent works \citep{Gielen:2013kla,Gielen:2013naa,Gielen:2014usa,Gielen:2014uga,Gielen:2014ila,Gielen:2015kua,Oriti:2016qtz,Oriti:2016ueo,deCesare:2016rsf}.

\subsection{Conservation laws in non-local field theories}

In local field theories there is a conserved current associated to
every continuous symmetry of the action given by the famous Noether
theorem. However, for non-local theories this result does not hold
as such, and must be generalized, due to the fact that the equations
of motion become integro-differential equations.

In \citep{Kegeles:2015oua} we derived an equivalent expression for
Noether currents for the case of non-local field theories, and for
the associated generalized conservation laws. In order to keep the
notation simple we present here a simplified version of the theorem,
referring to the original work for the full statement. 
\begin{thm}
If a non-local action $S=\int_{M}S^{L}+\int_{\tilde{M}}S^{I}$ is
symmetric under a group action generated by the vector fields $\left(X_{M},X_{V}\right)$
then the following identity holds for all $i$ 
\begin{align}
EL\left[X_{Q}\right] & =\sum_{c}\int D_{F}L^{I}\left(X_{cQ}\right)\left[\delta^{c}-\delta^{i}\right]\nonumber \\
 & -\text{Div}_{M}\left(D_{J}L^{L}\left[X_{Q}\right]+L^{L}\cdot X_{M}\right)\nonumber \\
 & -\text{Div}_{M}\left(\int_{\Omega}D_{J}L^{I}\left(X_{Q}^{c}\right)\delta^{c}\right)\label{eq:General non-local Noether theorem}\\
 & -\int_{\tilde{\Omega}}\text{Div}_{\tilde{M}}\left(L^{I}\cdot X_{\tilde{M}}\right)\delta^{i}.\nonumber 
\end{align}
Here $X_{\tilde{M}}$ denotes the vector field of base manifold transformations
of $\tilde{M}$ generated by $X_{M}$ as we discussed in the previous
section, $D_{F}L\left(X_{cQ}\right)$ denotes the Fréchet derivatives
of the Lagrangian in the direction of $X_{cQ}$, $\delta^{c}$ denotes
the delta distribution on the domain of the field of color $c$ and
the non-local Lagrangian $L^{I}=L^{I}\left(x,\phi\left(x\right),\partial\phi\vert_{x}\right)$
is assumed to be a function on the base manifold, fields at the point,
and first derivatives of the fields at the same point. The left hand
side denotes the equations of motion contracted with the vector field
$X_{Q}$.
\end{thm}
In the case when the non-local Lagrangian is independent of derivatives
of the fields, and the generators of the symmetry group of base manifold
transformations are divergence-free, $\text{div}\left(X_{M}\right)=0$,
and when the transformations of the field values is proportional to
the field value itself, $Q\left(\phi_{1,2,3}\right)\propto\phi_{1,2,3}$,
the above identity simplifies significantly to

\begin{align}
EL\left[X_{Q}\right] & =\Delta-\text{Div}_{M}\left(D_{J}L^{L}\left[X_{Q}\right]+L^{L}\cdot X_{M}\right),
\end{align}
where 
\begin{equation}
\Delta=\sum_{c}\int D_{cV}L^{I}\left(X_{Q}^{c}\right)\delta^{c},
\end{equation}
is referred to as correction term. This result explicitly shows that
the currents associated to symmetries of the non-local action are
no longer conserved. Instead their divergence are dictated by the
non-local part of the action.

After imposing equations of motion on the fields, we get the identity
that replaces the usual Noether theorem 
\begin{equation}
\text{Div}_{M}\left(D_{J}L^{L}\left[X_{Q}\right]+L^{L}\cdot X_{M}\right)=:\text{Div}\left(J\right)=\Delta.\label{eq:non-local Noether theorem}
\end{equation}
The quantity in brackets on the left hand side is the Noether current
of the local part of the action and the right hand side of the equation
is the non-vanishing divergence of the current due to non-local structure
of the theory.

It becomes now a straightforward calculation to apply the equation
\ref{eq:non-local Noether theorem} to models and symmetries introduced
in the previous section. In the rest of this section we summarize
the resulting identities. 

Since the local part of all our models is given by equation \ref{eq:local part},
the Noether current does not change and can be written as 
\begin{equation}
J=\kappa\sum_{c}\left(\nabla\phi^{c}\cdot\bar{X}_{Q}+\nabla\bar{\phi}^{c}\cdot X_{Q}\right)+L^{L}\cdot X_{M}.
\end{equation}
Note that for the $U\left(1\right)$ symmetry $X_{M}=0$, and the
Noether current becomes proportional to $\kappa$. This automatically
implies that for all static models the Noether current associated
to the $U\left(1\right)$ symmetry is zero. The correction term, however,
may not trivially vanish. Apart from the values for $\kappa$, the
models introduced earlier will differ only by the correction term
in equation \ref{eq:non-local Noether theorem}. In table \ref{tab:Models-and-their-correction terms}
we will present the correction terms for discussed models.

The notation in table \ref{tab:Models-and-their-correction terms}
is as follows. For brevity we do not indicate the base points and
write $\int\phi^{1}\phi^{2}\phi^{3}\phi^{4}$ in order to refer to
the non-local part of the Boulatov action. We also write $\phi^{1}\int\phi^{2}\phi^{3}\phi^{4}$
for 
\begin{equation}
\phi_{1,2,3}^{1}\int\text{d}g_{4,5,6}\,\phi_{3,4,5}^{2}\phi_{5,2,6}^{3}\phi_{6,4,1}^{4}.\label{eq:example of integral notation}
\end{equation}
The integral $\int\phi^{2}\phi^{3}\phi^{4}$ can be seen as a function
evaluated at the point $\left(g_{1},g_{2},g_{3}\right)$. We denote
the Lie derivative of this function with respect to the vector field
$X_{M}$ as $X_{M}\left(\int\phi\phi\phi\right)$. For brevity we
denote the expression in \ref{eq:example of integral notation} also
by the formal delta distribution 
\[
\int\phi^{1}\phi^{2}\phi^{3}\phi^{4}\,\delta^{1}=\phi_{1,2,3}^{1}\int\text{d}g_{4,5,6}\,\phi_{3,4,5}^{2}\phi_{5,2,6}^{3}\phi_{6,4,1}^{4},
\]
meaning that the integral over the domain of the field $\phi^{1}$
is to be excluded.

\section{Conserved charges in presence of matter}

In this section we will discuss the consequences of the matter coupling
introduced in \citep{Oriti:2016qtz} and show that such coupling implies
the existence of quantities which are constant in the matter field
variable, and can be interpreted as \textsl{conserved charges}.

While this could be taken as a fact of purely mathematical interest,
it may also indicate some underlying interesting physics, for quantum
gravity models. The reason is the following. The type of matter field
introduced in \citep{Oriti:2016qtz} was a free, massless, minimally
coupled real scalar field, entering as an additional variable in the
domain of the GFT fields, for 4d gravity models, whose classical dynamics
was then studied. The same classical dynamics was given an interpretation
as cosmological dynamics for continuum homogeneous universes, emerging
from the GFT system as quantum condensates. As customary in quantum
cosmology, and to some extent compulsory in background independent,
diffeomorphism invariant theories, the dynamics was expressed in terms
of relational observables \citep{Dittrich:2004cb,Dittrich:2005kc,Baratin:2010nn}.
In particular, the added scalar field was chosen to play the role
of relational clock, i.e. the physical time variable in terms of which
describing the evolution of all the geometric observables, e.g. the
volume of the universe. We refer to \citep{Oriti:2016qtz} for more
details. Remarkably, the same variable enters the GFT action just
as a standard, local time coordinate would enter an ordinary field
theory. This suggests a deeper physical meaning for the charges that,
following some symmetry of the corresponding GFT model, are in fact
conserved with respect to the same relational time variable/clock.
We do not discuss further the possible physical interpretation and
confine ourselves to the mathematical analysis of such extended models.

The domain of the GFT field is extended to become 
\begin{equation}
\phi:M\times\mathbb{R}\to\mathbb{C},
\end{equation}
where $M$ is the base manifold of the correspondent GFT model without
matter and $\mathbb{R}$ describes the degree of freedom of a real
scalar field. We call the new base manifold $M_{\text{mat}}=M\times\mathbb{R}$,
and denote a point on $M_{\text{mat}}$ as $\left(g_{1},\cdots,g_{n},\varphi\right)$.
The field value at this point is then denoted $\phi\left(g_{1},\cdots,g_{n},\varphi\right)=\phi_{1,\cdots,n,\varphi}$.
Intuitively we can think of GFT field $\phi$ as describing a ``chunk''
of space in which the scalar field takes the value $\varphi$.

The dynamics is then described by an action which is non-local in
the group variables, but local in the additional matter field variable.
This means that every Lagrangian in the non-local action sum is evaluated
at the same value $\varphi$: 
\begin{eqnarray}
S & = & \int_{M\times\mathbb{R}}\phi\left(\vec{g},\varphi\right)K\left(\vec{g},\varphi\right)\phi\left(\vec{g},\varphi\right)\nonumber \\
 & + & \int_{M\times\mathbb{R}}\phi\left(\vec{g},\varphi\right)\cdots\phi\left(\vec{h},\varphi\right)V\left(\vec{g},\cdots,\vec{h},\varphi\right).\quad,\label{eq:matter field action}
\end{eqnarray}
with the dependence of the various terms in the action on the additional
scalar field being motivated by an analysis of the GFT Feynman amplitudes
and their relation with simplicial gravity path integrals, and by
the required symmetries of the scalar field dynamics. The further
requirements that the scalar field is free, massless and minimally
coupled, plus some further approximation motivated by the hydrodynamics
setting \citep{Oriti:2016qtz}, lead to $K\left(\vec{g},\varphi\right)=\mathcal{K}\left(\vec{g}\right)+\Delta_{\varphi}$,
and to a vertex function $V$ that is independent of $\varphi$.

\subsection{Conserved charges and symmetries}

Locality in the matter field allows to define a local slicing with
respect to which we can construct conserved quantities $Q\left(\varphi\right)$
for \textsl{any} symmetry of the action, such that $\partial_{\varphi}Q\left(\varphi\right)=0$.
This is easily seen from the equation \ref{eq:General non-local Noether theorem},
where the integral domain is now replaced by $M_{\text{mat}}=M\times\mathbb{R}$
and the delta function $\delta^{c}$ that acts on the domain of the
field with color $c$ can be written as $\delta_{M}^{c}\delta_{\mathbb{R}}^{\varphi}$,
where $\delta_{M}^{c}$ acts on the group part of the domain and $\delta_{\mathbb{R}}^{\varphi}$
fixes the value of the matter field. Integrating the above equation
over $M$ and taking into account that the action is local in the
parameter $\varphi$, as well as the fact that the base manifold $M$
has no boundary\footnote{If the underling group of the model has a boundary, then boundary
terms need to taken into account.}, the above equation simplifies to 
\begin{align}
\int_{M}EL\vert_{\varphi}\left[X_{Q}\right] & \simeq\partial_{\varphi}\left(\partial_{\partial_{\varphi}\phi^{c}}L^{\text{loc}}\left[X_{Q}^{c}\right]+\partial_{\partial_{t}\bar{\phi}^{c}}L^{\text{loc}}\left[\bar{X}_{Q}^{c}\right]\right)\nonumber \\
 & +\partial_{\varphi}\left(L^{\text{loc}}\cdot X_{\varphi}\right)\nonumber \\
 & +\partial_{\varphi}\int_{M}\left(D_{\partial_{\varphi}\phi^{c}}L^{\text{int}}\left(X_{Q}^{c}\right)+D_{\partial_{\varphi}\bar{\phi}^{c}}L^{\text{int}}\left(\bar{X}_{Q}^{c}\right)\right)\nonumber \\
 & +\partial_{\varphi}\int_{M}\left(L^{\text{int}}\cdot X_{\varphi}\right),\label{eq:Charge equation}
\end{align}
where the equality is true up to a minus sign. Taking the $\varphi$
component of the current we get 
\begin{eqnarray}
Q\left(\varphi\right): & = & \partial_{\partial_{\varphi}\phi^{c}}L^{\text{loc}}\left[X_{Q}^{c}\right]+\partial_{\partial_{\varphi}\bar{\phi}^{c}}L^{\text{loc}}\left[\bar{X}_{Q}^{c}\right]+L^{\text{loc}}\cdot X_{\varphi}\nonumber \\
 & + & \int_{M}\left(D_{\partial_{\varphi}\phi^{c}}L^{\text{int}}\left(X_{Q}^{c}\right)+D_{\partial_{\varphi}\bar{\phi}^{c}}L^{\text{int}}\left(\bar{X}_{Q}^{c}\right)\right)\nonumber \\
 & + & \int_{M}\left(L^{\text{int}}\cdot X_{\varphi}\right).
\end{eqnarray}
Due to equation \ref{eq:Charge equation}, this satisfies on shell
\begin{equation}
\partial_{\varphi}Q\left(\varphi\right)=0.
\end{equation}
Since the interaction Lagrangian does not depend on derivatives of
$\phi$, the conserved charge becomes 
\begin{eqnarray}
Q\left(\varphi\right): & = & S\vert_{\varphi}\cdot X_{\mathbb{R}}^{\varphi}\label{eq:General charge equation}\\
 &  & +\int_{M}\left(\partial_{\partial_{\varphi}\phi^{c}}L^{\text{loc}}\left[X_{Q}^{c}\right]+\partial_{\partial_{\varphi}\bar{\phi}^{c}}L^{\text{loc}}\left[\bar{X}_{Q}^{c}\right]\right),\nonumber 
\end{eqnarray}
where $S\vert_{\varphi}=\int_{M}L^{\text{loc}}\vert_{\varphi}+\int_{M}L^{\text{int}}\vert_{\varphi}$
is the action in equation \ref{eq:matter field action} at a fixed
value of the parameter $\varphi$.

For example in the case of a $U\left(1\right)$ symmetry which is
generated by $X_{Q}=\imath\theta_{c}\phi^{c}$, with $\sum_{c}\theta^{c}=0$,
we get the conserved charge

\begin{equation}
Q\left(\varphi\right)=\imath\sum_{c}\theta^{c}\int_{M}\,\left(\phi^{c}\partial_{\partial_{t}\phi^{c}}L^{\text{loc}}-\bar{\phi}^{c}\partial_{\partial_{t}\bar{\phi}^{c}}L^{\text{loc}}\right).
\end{equation}

For the $SU\left(2\right)$ symmetry, with $X_{Q}=-X_{M}\left(\phi\right)$
and $X_{M}$ being left invariant generators of $SU\left(2\right)$
as in equation \ref{eq:left}, $Q$ takes instead the form 
\begin{eqnarray*}
Q\left(\varphi\right) & = & -\int_{M}\partial_{\partial_{\varphi}\phi^{c}}L^{\text{loc}}\left[X_{cM}\left(\phi^{c}\right)\right]\\
 &  & -\int_{M}\partial_{\partial_{\varphi}\bar{\phi}^{c}}L^{\text{loc}}\left[X_{cM}\left(\bar{\phi}^{c}\right)\right].
\end{eqnarray*}
This shows that we can easily calculate ``conserved'' quantities
for the symmetries we found earlier in the paper.

However, in addition to the symmetries on the group space we may also
have symmetries on $\mathbb{R}$ which correspond to symmetries of
the matter field, so the symmetry group of the models will be larger.

In general, the symmetry of the matter field will also be strongly
model dependent, and have to be investigated on a case by case basis.
However, in the case of free, massless, minimally coupled scalar matter,
the action is (and should be) invariant under matter field translations
of the form $\varphi\mapsto\varphi+\mu$. The charge for this symmetry
will take the following form 
\[
Q\left(\varphi\right)=-\int_{M}\left(\partial_{\partial_{\varphi}\phi^{c}}L^{\text{loc}}\partial_{\varphi}\phi^{c}+\partial_{\partial_{\varphi}\bar{\phi}^{c}}L^{\text{loc}}\partial_{\varphi}\bar{\phi}^{c}\right)+S\vert_{\varphi}.
\]
Defining $\Pi^{c}:=\partial_{\partial_{\varphi}\phi^{c}}L^{\text{loc}}$,
this takes the form of the Legendre transform of the Lagrangian defined
by $S\vert_{t}$ 
\begin{equation}
Q\left(\varphi\right)=-\int_{M}\left(\Pi^{c}\partial_{\varphi}\phi^{c}+\bar{\Pi}^{c}\partial_{\varphi}\bar{\phi}^{c}\right)+S\vert_{\varphi}.
\end{equation}

This is of course extremely suggestive of a GFT Hamiltonian with respect
to the evolution defined by the relational ``time'' $\varphi$,
and this is certainly an important point to be investigated further,
in both its mathematical and physical consequences.

It is important to note that there are very special conditions that
the matter field $\varphi$ has to satisfy to represent a good relational
clock. It is interesting to investigate further also how these conditions,
and their relaxation, reflect on the dependence of the GFT action
on the same matter field variable, and what field-theoretic consequences
they have, in particular concerning the existence and form of the
conserved charges we have found. Moreover, it is easy to realize that,
if the model has more than one matter field that enters the action
locally, the above treatment can be performed for any of the matter
fields. In this case, however, above equations will contain additional
boundary terms. We leave further analysis of these points to future
work.

\section{Conclusion and outlook}

In this paper we provided an extensive symmetry analysis for various
models in Group Field Theory.

We have elucidated the symmetry group of various GFT models, and how
it is affected by the various ingredients entering their definition:
rank, base group, color, combinatorial structure.

Our main result shows that, apart form the expected symmetry groups
of left multiplication and $U\left(1\right)$, the discussed models
do not posses any other continuous Lie point symmetries. This holds
even in the case of static, gauge invariant, models, in which the
Lagrangian does not depend on derivatives of fields. This is not obvious
since an ordinary local field theory without dynamical terms would
possess a fairly large gauge group of diffeomorphisms of the base
manifold. However, the presence of the interaction term with a particular
combinatorial structure as well as the requirement of gauge invariance
insures that the symmetry group becomes very small. In this sense
our treatment provides a complete set of point symmetries of discussed
models.

Using our previous result on conservation laws for non-local theories
we were then able to calculate generalized ``conservation'' laws
that correspond to continuous symmetries. And were able to show that
in particular cases of matter coupling to GFT fields our construction
provides a notion of conserved charges, the same way Noether theorem
does in local field theories. An existence of conserved quantities
shows, that once a matter field satisfies a notion of a ``good''
clock it also obtains the usual ``time'' properties in the field
theoretical frame-work. As we already mentioned, a lot more should
be understood about such conserved charges in GFT models.

It is an exciting and important task to understand the consequences
of the GFT symmetry groups on the physics of these models. This is
what needs primarily to be addressed in the future.

On the one hand, an understanding of conservation laws in terms of
geometrical objects could be a very important step in the development
of the theory. Conservation laws and conserved charge equations could
provide a field theoretical explanation of cosmological features stemming
from the underlying quantum gravity models, in the context of GFT
condensate cosmology \citep{Gielen:2016dss,Oriti:2016qtz,Oriti:2016ueo}.

The very existence of a condensate phase in GFT models, and more generally
their macroscopic phase diagram, currently being explored mainly by
FRG methods \citep{Benedetti:2014qsa,Geloun:2015qfa,Geloun:2016qyb},
can now be studied also on the basis of GFT symmetries and corresponding
symmetry breaking.

On the other hand a classification of symmetry groups in GFT could
be used as a better characterization of the theory space, a crucial
ingredient for systematic renormalization group studies \citep{BenGeloun:2011rc,Geloun:2011cy,Carrozza:2012uv,Carrozza:2013mna,Geloun:2016bhh}.

In particular this could help clarifying the connection between simplicial
and tensorial GFT models. As we noted, a further indication of their
close connection has been found already in our analysis, showing that
only colored GFT models of simplicial type appear to have an $U\left(1\right)$
symmetry as well as the unrestricted translation invariance that is
found in tensorial GFTs.

From a more mathematical point of view, it appears to be very interesting
to understand the extension of the symmetry groups we considered to
Lie-Baecklund or generalized symmetries, which requires a better characterization
of the equivalence class of GFT actions leading to the same classical
equations of motion.

Finally, we need to go beyond the purely classical analysis performed
in this paper, and move to the analysis of the same symmetries we
have identified at the quantum level, deriving and studying in detail
the corresponding Ward identities, and the issue of possible anomalies.

\newpage{}

\appendix

\section{Reduction of transformations due to gauge invariance \label{sec:Reduction-of-transformations}}

From equation \ref{eq:gauge invariance condition} the requirement
on the transformation reads 
\begin{equation}
C\left(\vec{g}h\right)=C\left(\vec{g}\right)\tilde{h}.
\end{equation}
Writing out this equation in components we get 
\begin{eqnarray}
C^{1}\left(g_{1}h\right) & = & C^{1}\left(g_{1}\right)\tilde{h}\nonumber \\
C^{2}\left(g_{2}h\right) & = & C^{2}\left(g_{2}\right)\tilde{h}\label{eq:Condition due to gauge invariance}\\
C^{1}\left(g_{3}h\right) & = & C^{1}\left(g_{3}\right)\tilde{h},\nonumber 
\end{eqnarray}
with $C^{i}$ being a diffeomorphism on the group $G$. At this point
we employ the known relation 
\begin{equation}
\text{Diff}\left(G\right)\simeq G\times\text{Diff}_{\mathds{1}}\left(G\right),
\end{equation}
that states that the group of diffeomorphisms on $G$ is diffeomorphic
(as a manifold) to the group $G$ itself (that acts by left multiplication)
times a group of diffeomorphisms that stabilizes the identity of $G$,
denoted $\text{Diff}_{\mathds{1}}\left(G\right)$. This implies that
every $C^{i}$ can be written by some $c^{i}\in G$ and $\mathcal{D}^{i}\in\text{Diff}_{\mathds{1}}\left(G\right)$
such that $C^{i}\left(g\right)=c^{i}\,\mathcal{D}^{i}\left(g\right)$
with $\mathcal{D}^{i}\left(\mathds{1}\right)=\mathds{1}$. Inserting
this relation in the equations \ref{eq:Condition due to gauge invariance}
and evaluating it a the point $g_{1}=g_{2}=g_{3}=\mathds{1}$ we observe
\begin{eqnarray}
c^{1}\cdot\mathcal{D}^{1}\left(h\right) & = & c^{1}\cdot\tilde{h}\\
c^{2}\cdot\mathcal{D}^{2}\left(h\right) & = & c^{2}\cdot\tilde{h}\\
c^{1}\cdot\mathcal{D}^{3}\left(h\right) & = & c^{1}\cdot\tilde{h},
\end{eqnarray}
which, in tern, implies 
\begin{equation}
\mathcal{D}^{1}\left(h\right)=\mathcal{D}^{2}\left(h\right)=\mathcal{D}^{3}\left(h\right)=\tilde{h}=:\mathcal{D}\left(h\right).
\end{equation}
Inserting this relation again in \ref{eq:Condition due to gauge invariance}
at an arbitrary point $\vec{g}$ we get for $\mathcal{D}$ 
\begin{equation}
\mathcal{D}\left(g_{i}h\right)=\mathcal{D}\left(g_{i}\right)\mathcal{D}\left(h\right).
\end{equation}
In other words $\mathcal{D}$ is an homomorphism and therefore an
automorphism. On $G$ however, the group of automorphisms splits in
the inner automorphisms which are given by a conjugation with a fixed
group element and outer automorphisms which are given by automorphisms
of the Dynkin diagram of the group and relate to the discreet symmetries.
Focusing on continuos transformations we get 
\begin{equation}
\mathcal{D}\left(g\right)=d\cdot g\cdot d^{-1}
\end{equation}
for some fixed $d\in SU\left(2\right)$.

\section{Barrett-Crane model \label{sec:Barrett-Crane-model}}

In this section we are going to show what are the admissible transformation
in the Barrett-Crane model.

In the following we will denote a group element of $Spin\left(4\right)$
by its two copies of $SU\left(2\right)$ a 
\[
Spin\left(4\right)\ni g=\left(g_{-},g_{+}\right),
\]
a tuple of four elements is referred to by the vector notation
\[
\vec{g}=\left(\vec{g}_{-},\vec{g}_{+}\right).
\]
We will also sometimes write $g_{1,2,3,4}$ for the tuple of elements
$\left(g_{1},g_{2},g_{3},g_{4}\right)$.

A base manifold transformation of the model is denoted by $C:Spin\left(4\right)^{\times4}\times SU\left(2\right)\to Spin\left(4\right)^{\times4}\times SU\left(2\right)$.
We denote the components of this transformation as 
\[
C\left(\vec{g},k\right)=\left(\left(C_{1}^{-},C_{1}^{+}\right),\cdots,\left(C_{4}^{-},C_{4}^{+}\right),C_{k}\right).
\]
Here all the component functions $C_{i}^{\pm}$ are functions on the
base manifold and therefore depend on points of the form $\left(\vec{g},k\right)$.
However, the combinatorial structure of the BC model dictates the
following conditions on the components 
\begin{align*}
C_{1}\left(g_{1,2,3,4},k_{1}\right) & =C_{4}\left(g_{10,8,5,1},k_{5}\right)\\
C_{2}\left(g_{1,2,3,4},k_{1}\right) & =C_{3}\left(g_{9,6,2,10},k_{4}\right).
\end{align*}
From the above relations we see that the components of the transformation
have the following dependences 
\[
C\left(g_{1,2,3,4},k\right)=\left(C_{1}\left(g_{1}\right),C_{2}\left(g_{2}\right),C_{3}\left(g_{3}\right),C_{4}\left(g_{4}\right),C_{k}\left(k\right)\right).
\]
A priori we do not have any additional constrains on the component
$C_{k}$. However, since $C$ is a diffeomorphism and $C_{i}$ are
diffeomorphisms, the transformation of the normal has to be a diffeomorphism
as well \footnote{Notice, that it would not be true if we didn't have restriction on
$C_{i}$, since then $C_{i}$ would not be a diffeomorphism and hence
neither needs to be $C_{k}$.}. Again invoking the diffeomorphism of manifolds $\text{Diff}\left(Spin\left(4\right)\right)\simeq Spin\left(4\right)\times\text{Diff}_{\mathds{1}}$
we denote the components of $C$ that belong to $\text{Diff}_{\mathds{1}}$
by the lower case $c$.

At this point we remind the reader that in the Barrett-Crane model
the gauge invariance of the fields was extended to incorporates simplicity
constrains 
\begin{eqnarray*}
\mathcal{S}: & \left(\vec{g};k\right) & \mapsto\left(\mathds{1};h_{-}^{-1}\right)\cdot\left(\vec{g};k\right)\cdot\left(\left(k\vec{u}k^{-1},\vec{u},\right);\mathds{1}\right)\cdot\left(h;h_{+}\right).
\end{eqnarray*}
Where $\cdot$ stands for the group multiplication and $;$ separates
the $Spin\left(4\right)$ elements from $SU\left(2\right)$. This
means that the fields of the model are invariant under $\mathcal{S}$,
\[
\phi\circ\mathcal{S}=\phi.
\]
Since the fields are transformed under $C$ as $\phi\mapsto\phi\circ C^{-1}$
we again get the following relations for the transformation $C$ 
\[
\phi\circ C\circ S=\phi\circ C,
\]
or equivalently for each $h\in Spin\left(4\right)$, $u\in SU\left(2\right)^{\times4}$
and $g_{i}\in Spin\left(4\right)$ there exist $\tilde{u}\in SU\left(2\right)^{\times4}$
and $\tilde{h}\in Spin\left(4\right)$ and $\tilde{k}\in SU\left(2\right)$
such that 
\begin{eqnarray}
C_{i}\left(g\cdot u_{k}\cdot h\right) & = & C_{i}\left(g\right)\cdot\tilde{u}_{C_{k}}\cdot\tilde{h}\label{eq:1}\\
C_{k}\left(h_{-}^{-1}kh_{+}\right) & = & \tilde{h}_{-}^{-1}C_{k}\left(k\right)\tilde{h}_{+},\label{eq:2}
\end{eqnarray}
where we write $u_{k}=\left(kuk^{-1},u\right)$. It is again obvious
that the left multiplication by $Spin\left(4\right)$ is untouched
by this transformation, however this is not true for normal component
$C_{k}$. We first focus on the transformations $C_{i}$ and treat
the normal component $C_{k}$ afterwards.

From the from of $u_{k}$ we notice that for $u=\mathds{1}$ the left
hand side does not depend on $k$ and so should't the right hand side.
It follows that for $u=\mathds{1}$ we have $\tilde{u}=\mathds{1}$.
Equation \ref{eq:1} then reads for the $\text{Diff}_{\mathds{1}}$
part, 
\[
c_{i}\left(g\cdot h\right)=c_{i}\left(g\right)\cdot\tilde{h}.
\]
It follows that $c_{i}$ is a homomorphism on $Spin\left(4\right)$
and therefore is either conjugation by a fixed element of $Spin\left(4\right)$
or a flip of the $SU\left(2\right)$ parts, which is a discrete transformation.
Hence, if $c_{i}$ is continuous it can be written as 
\[
c_{i}\left(g\right)=s\cdot g\cdot s^{-1},
\]
where $g,s\in Spin\left(4\right)$. This implies 
\[
\tilde{h}=s\cdot h\cdot s^{-1}.
\]
Inserting this relation now in equation \ref{eq:2} we obtain 
\[
C_{k}\left(h_{-}^{-1}kh_{+}\right)=\left(s_{-}h_{-}^{-1}s_{-}^{-1}\right)\,C_{k}\left(k\right)\,\left(s_{+}h_{+}s_{+}^{-1}\right).
\]
Splitting $C_{k}$ in the left multiplication by $SU\left(2\right)$
and $\text{Diff}_{\mathds{1}}$ we get for some fixed $w\in SU\left(2\right)$
\begin{equation}
w\,c_{k}\left(h_{-}^{-1}kh_{+}\right)=\left(s_{-}h_{-}^{-1}s_{-}^{-1}\right)w\,c_{k}\left(k\right)\,\left(s_{+}h_{+}s_{+}^{-1}\right).\label{eq:3}
\end{equation}
Choosing $h_{-}=h_{+}$ and setting $k=\mathds{1}$ we get 
\[
w=\left(s_{-}h_{-}^{-1}s_{-}^{-1}\right)w\,\left(s_{+}h_{-}s_{+}^{-1}\right),
\]
which can only be satisfied if $w=\mathds{1}$.

Inserting equation \ref{eq:3} in \ref{eq:1} and using the fact that
$c_{i}$ is a homomorphism yields 
\begin{eqnarray*}
c_{i}\left(u_{k}\right) & = & c_{i}\left(k,\mathds{1}\right)\cdot c_{i}\left(u,u\right)\cdot c_{i}\left(k^{-1},\mathds{1}\right)\\
 & = & c_{k}\left(k\right)c_{i}\left(u,u\right)c_{k}\left(k^{-1}\right).
\end{eqnarray*}
Hence, $c_{i}\left(a,b\right)=\left(c_{k}\left(a\right),c_{k}\left(b\right)\right)$
and $c_{k}$ is a homomorphisms itself. Therefor 
\[
c_{i}\left(g\right)=\left(s,s\right)\cdot g\cdot\left(s,s\right)^{-1},
\]
and $c_{k}\left(k\right)=sks^{-1}$.

These are the only admissible transformations that preserve the combinatorial
structure of the theory and respect the simplicity constraints together
with gauge invariance. Notice that $\mathcal{S}$ itself would fail
the requirement \ref{eq:condition on the group action} and therefor
can not be seen as a base manifold transformations, which is why we
dot not obtain the symmetry under $\mathcal{S}$ in this approach.

\section{Constancy of the phase\label{sec:Constancy-of-the}}

In this section we are going to prove the following statement 
\begin{thm*}
If for any point $g_{1},\cdots,g_{6}\in G$ where $G$ is a simple
Lie group the following equation holds 
\[
\sum_{i}^{4}\theta^{i}\left(\vec{g}_{i}\right)=0.
\]
And for any $h\in G_{D}\left(2\right)$ the functions $\theta^{i}$
satisfy 
\begin{equation}
\theta^{i}\circ R_{h}=\theta^{i},\label{eq:gi}
\end{equation}
then the functions $\theta^{i}$ are constants that add up to zero,
$\theta^{1}+\theta^{2}+\theta^{3}+\theta^{4}=0$ . 
\end{thm*}
We first prove the following lemma 
\begin{lem}
Let $\theta$ be a function from a Lie group $G$ to $\mathbb{R}$
such that for all $g\in G$ the difference 
\[
\theta\left(gh\right)-\theta\left(g\right)=f\left(h\right)
\]
is a function only on the ``distance'' of the points $h$. Then
$f$ is a homomorphism from the group $G$ to $\left(\mathbb{R},+\right)$. 
\end{lem}
\begin{proof}
From the definition it follows that $f\left(\mathds{1}\right)=0$.
Choosing $g=\tilde{g}h^{-1}$ we get 
\[
f\left(h\right)=\theta\left(gh\right)-\theta\left(g\right)=\theta\left(\tilde{g}\right)-\theta\left(\tilde{g}h^{-1}\right)=-f\left(h^{-1}\right).
\]
Choosing $g=gh\tilde{h}$ we also get 
\begin{eqnarray*}
f\left(h\tilde{h}\right) & = & \theta\left(gh\tilde{h}\right)-\theta\left(g\right)\\
 & = & \theta\left(gh\tilde{h}\right)\pm\theta\left(gh\right)-\theta\left(g\right)\\
 & = & f\left(\tilde{h}\right)+f\left(h\right).
\end{eqnarray*}
Which concludes the proof. 
\end{proof}
We now prove the above theorem. 
\begin{proof}
The above equation then reads 
\begin{equation}
\theta^{1}\left(\vec{g}_{1}\right)+\theta^{2}\left(\vec{g}_{2}\right)+\theta^{3}\left(\vec{g}_{3}\right)+\theta^{4}\left(\vec{g}_{4}\right)=0,\label{eq:sum of phases}
\end{equation}
where $\vec{g}_{1}=\left(g_{1},g_{2},g_{3}\right)$, $\vec{g}_{2}=\left(g_{3,}g_{4},g_{5}\right)$,
$\vec{g}_{3}=\left(g_{5},g_{2},g_{6}\right)$ and $\vec{g}_{4}=\left(g_{6},g_{4},g_{1}\right)$.
Than for any differentiable curve $c:\mathbb{R}\supset I\to SU\left(2\right)$
with $c\left(0\right)=\mathds{1}$ the above equation is true if we
replace $g_{1}$ by the curve $c\left(t\right)$. Deriving the resulting
equation with respect to the parameter $t$ we get 
\[
\partial_{t}\theta^{1}\left(c\left(t\right),g_{2},g_{3}\right)+\partial_{t}\theta^{4}\left(g_{6},g_{4},c\left(t\right)\right)=0.
\]
By integration we obtain 
\[
\theta^{1}\left(c\left(t\right),g_{2},g_{3}\right)-\theta^{1}\left(\mathds{1},g_{2},g_{3}\right)=\theta_{1}^{1}\left(c\left(t\right)\right),
\]
for some function $\theta_{1}^{1}\left(c\left(t\right)\right)$. Applying
the same argument to $\theta^{4}$ we gain the following relations,
\begin{eqnarray*}
\theta^{1}\left(g_{1},g_{2},g_{3}\right) & = & \theta_{1}^{1}\left(g_{1}\right)+\theta^{1}\left(\mathds{1},g_{2},g_{3}\right)\\
\theta^{4}\left(g_{6},g_{4},g_{1}\right) & = & -\theta_{1}^{1}\left(gd_{1}\right)+\theta^{4}\left(\mathds{1},g_{2},g_{3}\right).
\end{eqnarray*}
Inserting these relations into equation \ref{eq:sum of phases} yields
\[
\theta^{1}\left(\mathds{1},g_{2},g_{3}\right)+\theta^{2}\left(\vec{g}_{2}\right)+\theta^{3}\left(\vec{g}_{3}\right)+\theta^{4}\left(g_{6},g_{4},\mathds{1}\right)=0.
\]
Successively performing the same step for all other group elements
$g_{i}$ eventually leads to the separation of the functions $\theta^{i}$
as follows, 
\begin{equation}
\theta^{i}\left(g_{1},g_{2},g_{3}\right)=\theta_{1}^{i}\left(g_{1}\right)+\theta_{2}^{i}\left(g_{2}\right)+\theta_{3}^{i}\left(g_{3}\right)+\text{const}.^{i},\label{eq:expression ofr theta 1}
\end{equation}
where $\theta_{j}^{i}$'s satisfy 
\begin{align*}
\theta_{1}^{1} & =-\theta_{3}^{4} & \theta_{2}^{1} & =-\theta_{2}^{3} & \theta_{3}^{1} & =-\theta_{1}^{2}\\
 &  & \theta_{2}^{2} & =-\theta_{2}^{4} & \theta_{3}^{2} & =-\theta_{1}^{3}\\
 &  &  &  & \theta_{3}^{3} & =-\theta_{1}^{4}.
\end{align*}
Using the requirement on gauge invariance (equation \ref{eq:gi})
yields for any $h\in G_{3D}$ 
\[
\theta_{1}^{i}\left(g_{1}h\right)+\theta_{2}^{i}\left(g_{2}h\right)+\theta_{3}^{i}\left(g_{3}h\right)=\theta^{i}\left(\vec{g}\right).
\]
Since this equation needs to hold for any $\vec{g}\in G^{\times3}$
we get for each $\theta_{j}^{i}$ the following condition 
\begin{equation}
\theta_{j}^{i}\left(gh\right)-\theta_{j}^{i}\left(g\right)=f_{j}^{i}\left(h\right),\label{eq:homo}
\end{equation}
with some functions $f_{j}^{i}$ . From the above lemma it follows
that $f_{j}^{i}$ is a homomorphism from $G$ to $\left(\mathbb{R},+\right)$.
Since $\left(\mathbb{R},+\right)$ is abelian and $G$ is simple $f$
is a constant zero function, $f=0.$

Evaluating equation \ref{eq:homo} at $g=\mathds{1}$ proves that
\begin{align*}
\theta_{j}^{i} & =\text{const}^{i},
\end{align*}
which together with equation \ref{eq:expression ofr theta 1} proves
\[
\theta^{i}\left(\vec{g},\phi^{c}\right)=\theta^{i},
\]
for some constants $\theta^{i}$. The conditions on the constants
follows. 
\end{proof}
\bibliographystyle{kp}
\bibliography{/Users/Alexander/Physik/Gravity/References/References}

\end{document}